\documentclass{article}
\usepackage[utf8]{inputenc}
\usepackage{amsmath,amssymb,amsfonts,amsthm,amscd}
\usepackage{graphicx}
\usepackage{color}
\usepackage{comment}
\usepackage[T1]{fontenc}
\usepackage{authblk}
\usepackage{amssymb}
\usepackage{amsmath}
\usepackage{caption}
\usepackage{subcaption}
\usepackage{comment}
\usepackage{pst-node}
\usepackage{tikz-cd} 
 \usepackage{color}
\newcommand{\AM}{\textcolor{black}}
\newcommand{\MC}{\textcolor{black}}

\usepackage{cancel}

\newtheorem{lemma}{Lemma}[section]

\newtheorem{proposition}{Proposition}[section]
\newtheorem{remark}{Remark}[section]
\newtheorem{example}{Example}[section]

\def\w{\omega}

\title{\AM{A reduction scheme for} coupled Brownian harmonic oscillators}

\author[1]{Matteo Colangeli\thanks{matteo.colangeli1@univaq.it}}
\author[2]{Manh Hong Duong\thanks{h.duong@bham.ac.uk}}
\author[3]{Adrian Muntean\thanks{adrian.muntean@kau.se}}
\affil[1]{Department of Information Engineering, Computer Science and Mathematics,
University of L'Aquila, Italy.}
\affil[2]{School of Mathematics,
University of Birmingham,
UK.}
\affil[3]{Department of Mathematics and Computer Science \& Centre for Societal Risk Research (CSR), Karlstad University, Sweden.}

\begin{document}

\maketitle
\begin{abstract}
\AM{We propose a} reduction scheme for a system \MC{constituted by} two coupled \AM{harmonically-bound} Brownian oscillators.
\MC{We reduce the description by constructing a lower dimensional model which inherits some of the basic features of the original dynamics and is written in terms of suitable transport coefficients.} %\AM{Essentially, we wish to reduce the coupled evolution system to a single equation possibly of the same type as the original dynamics, where the  effective coefficients are computable.} 
\MC{The proposed procedure is twofold}: while the deterministic component of the dynamics is obtained \AM{by a direct application of} the invariant manifold method, the diffusion terms are determined via the Fluctuation-Dissipation Theorem.
\AM{We highlight} the behavior of the coefficients %as functions of the coupling parameter and also \AM{discuss} the onset of a critical value \AM{of the latter. 
up to a critical value of the coupling parameter, which marks the endpoint of the interval in which a contracted description is available.
The study of the weak coupling regime is addressed and the commutativity of alternative reduction paths is also discussed. 
%The proposed method is also applicable to more complex scenarios, involving e.g. a large set of coupled Brownian oscillators. 
\end{abstract}

\section{Introduction}
Complex systems in nature and in applications (such as molecular systems, crowd dynamics, swarming, opinion formation, just to name a few) are often described by systems of coupled stochastic differential equations modeling the time evolution of a large number of microscopic constituents.
While a fundamental approach to the analysis of a many-particle system would require solving the entire set of coupled Newton equations equipped with the appropriate initial conditions, a stochastic approach is conveniently exploited when the number of degrees of freedom becomes large. This way, only the dynamics of few selected constituents is addressed, while the details of the microscopic interactions with the unresolved particles are blurred and encoded in a noise term with prescribed statistical properties.
%stochastically interacting particle systems. 
\MC{Such stochastic models are thus expected to convey a faithful description of the underlying physical phenomena.} 
%provide a detailed description, often posed at a more microscopic level, of the underlying phenomena. 
\MC{However, stochastic dynamics can still be hard to treat analytically, and even their numerical characterization may not be easily accessible, owing to the high dimensionality of the model (degrees of freedom, number of involved parameters, etc.).}
%However, even dealing  with the full models \MC{may be} usually analytically hard and computationally expensive due to their high dimensionality (degrees of freedom, number of involved parameters) and complexity.
It is thus of \AM{great} importance to approximate such large and complex systems by simpler and lower dimensional ones, while still preserving the essential information \MC{from the original model}. Such research endeavour has gradually developed and systematized into the growing field of model reduction \MC{methods} \cite{gorban2006model,Ott05,Snowden,Rupe}. A natural reduction approach consists in keeping track of only a set of relevant, typically slow, variables (for instance, certain marginal coordinates), also referred to as \textit{collective variables} in the literature \cite{hartmann}.
%which are called collective variables \MCb{(question: does then name ``collective variables'' belong to an established jargon? or are we introducing it here in this manuscript?)}\HD{I think it is quite common in the coarse-graining literature} or reaction coordinates, 
%\MC{and to ignoring the dynamics of the unresolved ones}. 
By projecting the full \MC{Markovian} dynamics \MC{on the manifold parametrized by} the relevant variables, a non-Markovian dynamics is typically obtained \cite{Zwanzig}. %non-Markovian (non-closed) system \cite{Zwanzig}, in the sense that the resulting system still requires information from the original dynamics \MCb{(why ``non-closed''? not clear! actually, the simplest example of a Brownian harmonic oscillator discussed at p. 20 of Ref. \cite{Zwanzig} shows that the elimination of the velocity variable leads to an exact \textbf{closed} non-Markovian equation for the position variable only.)}. 
%Therefore, a
A key step in the reduction procedure consists, then, in %approximate the non-Markovian terms by Markovian ones obtaining simpler closed-form models 
\MC{verifying the conditions which allow to restore the Markovian structure of the original process, so as to avoid the onset of memory terms in the dynamics.}
%\MCb{(remark: again, the foregoing discussion about ``closed-non closed'' models looks quite fuzzy. If it is relevant to the general discussion, it should be rephrased a bit more clearly.)}. 
To this aim, many different approaches have been proposed in the literature. Most of the existing works consider systems in which a time-scale separation is explicitly invoked, see e.g. the monographs \cite{pavliotis2008multiscale, gorban2006model} for further information. In such systems characterized by distinguished time-scales, as the ratio of fast-to-slow characteristic time-scales increases, the fast dynamics equilibrates more and more rapidly with respect to the the slow ones. Consequently,  a fully \MC{Markovian} reduced description for the evolution of the slow variables can be derived \MC{in the limit of a perfect time-scale separation}, see e.g. \cite{Zwanzig,Ghil}. 

A classical example, which has been thoroughly studied in the literature, cf. e.g. \cite{Kramers1940,Nelson1967,Duong2017} and references therein, deals with the derivation of the first order overdamped Langevin dynamics from the second-order underdamped model in the zero-mass limit (or the high friction limit, also known as the Smoluchowski-Kramers approximation), which corresponds to sending the mass to zero (or the friction coefficient to infinity, respectively). A major limitation caused by the ansatz of perfect time-scale separation is that a large wealth of information is lost in the limiting process. Specifically, inertial effects are essentially disregarded in the Smoluchowski dynamics. It is thus desirable to develop \MC{techniques of model reduction  not relying on the ansatz of perfect time-scale separation}. Several efforts in this direction have been made in recent years, see e.g. \cite{TonyRoberts}. For instance, in \cite{Wouters2019,Wouters2019b} the authors outline a reduction procedure, based on the Edgeworth expansion, designed for both deterministic and stochastic systems featuring a moderate time-scale separation.
Moreover, the analysis of Ruelle–Pollicott resonances in a reduced state space in the presence of a weak time-scale separation has been recently discussed in \cite{Checkroun1,Checkroun2}.
In \cite{Lu2014} the authors obtain closed reduced models for the overdamped Langevin dynamics for a specific class of collective variables (reaction coordinates) that satisfy a suitable Poisson equation. The papers~\cite{Legoll2010,zhang2016effective,Duong2018, Legoll2019,Lelievre2019,Hartmann2020} deal with more general diffusion processes and collective variables using conditional expectations (see also \cite{SharmaHilder2022} for a similar work done in the context of Markov chains). Recently, in \cite{ColMun22}, the first and third authors of the present paper introduced a new approach to derive a closed reduced model for the underdamped Langevin dynamics based on the use of the dynamics invariance principle \cite{GorKar05} for the deterministic part and of the fluctuation-dissipation theorem for the noise term. The reduced dynamics not only provides a meaningful correction to the classical Smoluchowski equation but also preserves the response function of the original dynamics in the high-friction limit.

We also remark that reduction procedures for models related to ours have been frequently investigated in the literature. For instance, in \cite{TakashiUneyama2019} the authors derive a system of overdamped Langevin equations starting from the dynamics of two coupled Brownian oscillators, while in \cite{Soheilifard2011} the authors consider a set of coupled overdamped Langevin equations and derive, using an alternative method, a reduced description for just one of the two oscillators. 
%The foregoing papers, which use different methods, thus look at different steps of the reduction path drawn in Fig. \ref{fig: diagram}. 
Remarkably, our procedure allows us to encompass the various steps of reduction by just reiterating the same scheme: this clearly stands as a noteworthy feature of our approach. 

The rest of the paper is organized as follows. In Section \ref{sec:model} we present the model describing the dynamics of two coupled Brownian oscillators and we illustrate the alternative reduction paths that will be considered next.
In Sec. \ref{sec:under} we discuss a first application of the reduction scheme, leading to the dynamics of a single Brownian oscillator, in the presence of inertial terms. A second application comes in Sec. \ref{sec:IM2}, where we discuss the derivation of an overdamped dynamics for two coupled Brownian oscilators.
The derivation of a further reduced dynamics in terms of a single overdamped Brownian oscillator is then outlined in Sec. \ref{sec:commute}, where we also show that the different reduction paths considered in the previous Sections commute. Conclusions are finally drawn in Sec. \ref{sec: conclusion}. Detailed technical computations are deferred to the Appendix.

\section{The model}
\label{sec:model}

The aim of this paper is to extend the methodology developed in \cite{ColMun22} and make it applicable to a more complex setting involving systems of coupled stochastic differential equations. To fix the ideas, we consider a model describing isolated elongated polymers whose dynamics can be approximated fairly well by the time evolution of two beads of equal mass $m$, connected together via a spring. We denote the positions of the two beads by $x_1$ and $x_2$, respectively. Balancing the exerted forces, we deduce the following system of Langevin equations:
%of the beads result from the balance of different forces expressed as
\begin{align}
\ddot{x}_1&=-\gamma_1 \dot{x}_1-\frac{1}{m}\nabla U_1(x_1)+\frac{1}{m}\nabla F(x_2-x_1)+\sigma_1\dot{W}_1\;,
\label{x1} \\
\AM{\ddot{x}_2}&=-\gamma_2 \dot{x}_2-\frac{1}{m}\nabla U_2(x_2)-\frac{1}{m}\nabla F(x_2-x_1)+\sigma_2\dot{W}_2 \label{x2} \;.
\end{align}
Here $W_i$, $i=1,2$, are standard independent  Brownian motions\footnote{To simplify the notation, in the sequel we denote by $\dot{W}$ the formal derivative of a Wiener process, corresponding to a white noise signal.}; $\gamma_i$ and $\sigma_i=\sqrt{2\gamma_i/\beta m}$ denote, respectively, the friction coefficients and the strength of the stochastic noise for the $i$-th bead \AM{($i\in\{1,2\}$)}, with $\beta$ the inverse temperature; $-\nabla U_i(x)$ is the external force acting on the $i$-th bead, while $\pm\nabla F(x_2-x_1)$ denote the interaction forces, which, according to Newton's third law (\AM{action-reaction principle}), have opposite signs and depend only on the magnitude of end-to-end vector $|z|=|x_2-x_1|$. We set $m,\gamma_i,\sigma_i\in (0,+\infty)$ for any $i\in\{1,2\}$. The above system was introduced in \cite{Degond2009} as a kinetic model for polymers endowed with inertial effects. 

Let $v_i=\dot{x}_i$ denote the velocity of the $i$-th bead. The system above can hence be rearranged as a first-order system:
\begin{align*}
\dot{x}_1&=v_1,
\\\dot{v}_1&=-\frac{1}{m}\nabla U_1(x_1)+\frac{1}{m}\nabla F(x_2-x_1)-\gamma_1 v_1+\sigma_1\dot{W}_1,
\\ \dot{x}_2&=v_2,
\\ \dot{v}_2&=-\frac{1}{m}\nabla U_2(x_2)-\frac{1}{m}\nabla F(x_2-x_1)-\gamma_2 v_2+\sigma_2\dot{W}_2.
\end{align*}
\AM{We restrict our attention exclusively to the one-dimensional case, however the working technique is not dependent on the space dimension. Furthermore, we employ in our description only Hookean springs, i.e. we involve  quadratic potentials of type}
\begin{equation}
U_i(x_i)=\frac{1}{2}\kappa_{i} x_i^2  \; ,
\end{equation}
with forcing terms originating from 
\begin{equation}
F(z)=\frac{1}{2}\xi z^2 \;.
\end{equation}
\AM{Here $\kappa_1,\kappa_2,\xi$ are given strictly positive  constants.} 

\AM{The original dynamics cf. } Eqs. \eqref{x1}-\eqref{x2} can thus be cast into the following system of \AM{linearly} coupled Langevin equations:

\begin{align}
\dot{x}_1&=v_1 \;, \nonumber\\
\dot{v}_1&=-\w_1^2 x_1+k(x_2-x_1)-\gamma_1 v_1+\sigma_1\dot{W}_1 \,, \label{original}\\ 
\dot{x}_2&=v_2 \;, \nonumber\\ \dot{v}_2&=-\w_2^2 x_2-k(x_2-x_1)-\gamma_2 v_2+\sigma_2\dot{W}_2 \;,\nonumber
\end{align}
where we introduced the frequencies $\omega_i=\sqrt{\kappa_i/m}$ and the coupling parameter $k=\xi/m$. This specific system has been used extensively in physics, biology and chemistry, for instance in the modelling of linear electrical networks \cite{Wang1945}, and to study the dynamics of atoms in protein molecules \cite{Berkowitz1981}.
%\subsection*{Our reduction paths}

Let us now illustrate two different reduction paths for the system of linearly coupled Langevin equations \eqref{original}, see the diagram in Fig. \ref{fig: diagram} for an illustration. 

\begin{figure}
\centering

%\[
%\begin{tikzcd}[row sep=8em,column sep=8em]
%\text{coupled dynamics for}~ \MC{\{x_1,v_1,x_2,v_2\}}
%(x_1, v_1)~\text{and}~(x_2,v_2) 
%\arrow{r}{\text{eliminating}~ (x_2,v_2)} \arrow[swap]{d}{\text{eliminating}~(v_1,v_2)} & \text{reduced dynamics for}~\MC{\{x_1,v_1\}} \arrow{d}{\text{eliminating}~v_1} \\
%\text{reduced dynamics for}~\MC{\{x_1,x_2\}} \arrow{r}{\text{eliminating}~x_2} & \text{reduced dynamics for}~\MC{\{x_1\}}
%\end{tikzcd}
%\]
\includegraphics[width=0.8\textwidth]{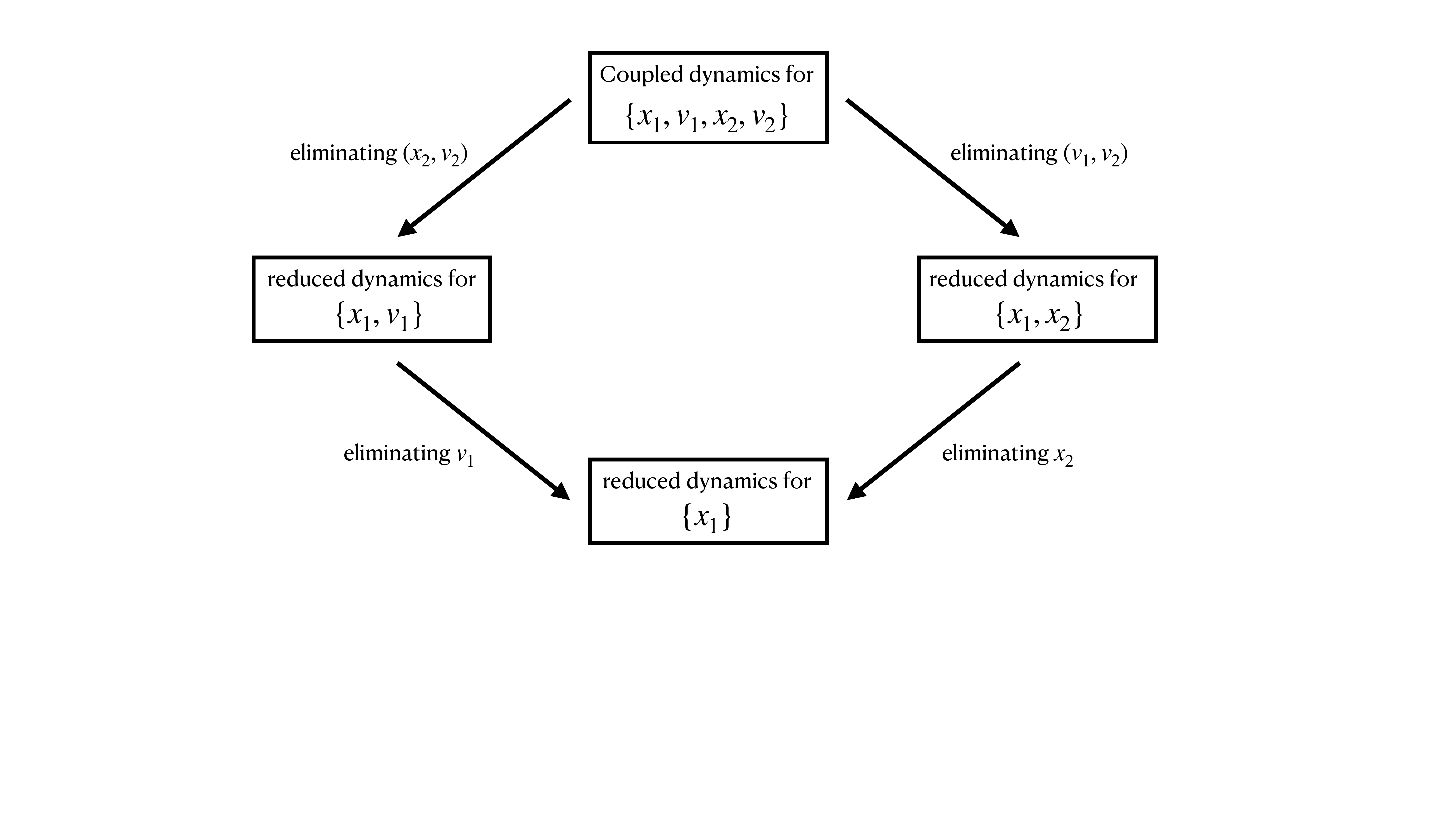}
\caption{The stairs of reduction for the coupled Langevin dynamics in Eq. \eqref{original}. Along the first path (top left and bottom left arrows), one first eliminates the variables $(x_2,v_2)$ of the second oscillator and then the velocity $v_1$ of the first oscillator. Instead, along the second route (top right and bottom right arrows) one first removes the two velocities $(v_1,v_2)$ and then, finally, the position $x_2$ of the second oscillator. The diagram is shown to be commutative: the two alternative paths lead to the same Langevin dynamics for the variable $x_1$.}
\label{fig: diagram}
\end{figure}

The first path consists in obtaining a contracted description expressed in terms of the variables $\{x_1,v_1\}$ only, thus formally erasing the variables describing the second oscillator, $(x_2,v_2)$ (see the top left arrow in Fig. \ref{fig: diagram}). In the second path, which is motivated from the classical Smoluchowski-Kramers approximation for a single Langevin dynamics, we instead eliminate the velocities $(v_1,v_2)$ to obtain a reduced model described by the only variables $\{x_1,x_2\}$ (top right arrow in Fig. \ref{fig: diagram}). Using the same approach, we may eventually contract the description even further, by removing the velocity $v_1$ in the first reduced model (bottom left arrow in Fig. \ref{fig: diagram}) or the position $x_2$ in the second one (bottom right arrow). We will then show that the both reduction paths lead to the same reduced dynamics for the remaining position variable $x_1$. Namely, the reduction diagram portrayed in Fig. \ref{fig: diagram} is commutative.

Along each step of the reduction path, we employ the invariant manifold method \cite{GorKar05} to derive lower dimensional models for the deterministic components of the original dynamics. The method was originally introduced as a special analytical perturbation technique in the KAM theory of integrable Hamiltonian systems \cite{Arn,Kol,Mos} and was later exploited in the kinetic theory of gases to derive the evolution equations of the hydrodynamic fields from the Boltzmann equation or related kinetic models \cite{GORBAN201848,GorKar05,GorKar}. Using the invariant manifold method, we succeed to derive a set of algebraic equations, called Invariance Equations, whose solutions fully characterize the coefficients of the reduced models. The same equations can also be solved in a perturbative fashion via the Chapman-Enskog method, which allows us to identify the structure of the so-called weak coupling regime. 
%We also provide criteria to select the relevant solutions of the Invariance Equations, by requiring that the eigenvalues of the reduced dynamics recover the eigenvalues of the original one as the coupling parameter tends to $0$.  
We then invoke the Fluctuation-Dissipation Theorem \cite{Kubo66,Kubo}, as well as the scale invariance of the stationary covariance matrix, to incorporate the noise terms in the reduced description. More details will come in the next Sections.
%The Fluctuation-Dissipation relation enables us to determine the diffusion matrix by solving a Lyapunov equation. 

%The latter involves a stationary matrix that characterizes the long-time limit of the covariance matrix associated to the dynamics. We require that the stationary matrix for the reduced dynamics coincides with the corresponding matrix evaluated through the original dynamics. Thus our approach automatically preserves important features of the original dynamics: the eigenvalues of the transport (deterministic) part, the stationary matrix and the fluctuation-dissipation relationship.

%Our analysis also reveals the onset of an interesting bifurcation phenomenon and the existence of a critical value for the coupling parameter at which the eigenvalues of the transport matrix become complex valued. This phenomenon is reminiscent of the well-known 1D13M Grad system studied in \cite{colan07}.

%\subsection*{Organization} The rest of the paper is organized as follows. In Section \ref{sec:under} we discuss a first application of the reduction scheme to Eq. \eqref{original}, corresponding to the top left arrow in the diagram shown in Fig. \ref{fig: diagram}. A second application, represented by the top right arrow in the same diagram, comes in Sec. \ref{sec:IM2}, along with the analysis of further reduction steps (the bottom left and right arrows in the diagram) which allows us to prove the commutativity of the reduction diagram. Conclusions are drawn in Sec. \ref{sec: conclusion}. Finally, detailed technical computations are deferred to the Appendix.

\section{Reduced dynamics for a single Brownian oscillator}
\label{sec:under}
%There are a couple of distinct ways to obtain a reduced description of the model of two coupled Brownian oscillators. In this section, we start by reducing the number of variables by simply removing one of the oscillators from the description. 
%Ideally we wish to formally get rid of the dynamics of the second oscillator from the reference dynamics \eqref{original}. Instead, we will need to provide effective coefficients for the remaining oscillators such that the missing description is compensated.
%\subsection{Method of the invariant manifold}
\label{sec:IM}
Let us denote by $\langle \mathcal{O} \rangle$ the conditional average over noise of the observable $\mathcal{O}$, subject to a prescribed deterministic initial datum.
The original dynamics \eqref{original} can thus be formulated as a linear system of ODEs \AM{as follows}:
\begin{equation}
    \dot{\mathbf{z}}=\mathbf{Q} \mathbf{z}  \, , \label{original2}
\end{equation}
where $\mathbf{z}=(\langle x_1 \rangle,\langle v_1 \rangle,\langle x_2 \rangle,\langle v_2 \rangle)$, and
\begin{equation}
    \mathbf{Q}= \mathbf{Q}(k)=\begin{pmatrix}
0 & 1 & 0 & 0 \\
-\omega_1^2-k & -\gamma_1 & k & 0 \\
0 & 0 & 0 & 1 \\
k & 0 & -\omega_2^2-k & -\gamma_2 
\end{pmatrix} \, . \label{Q}
\end{equation}
%with $k\ge 0$ denoting the coupling parameter between the two oscillators dynamics.

The following result, whose proof is deferred to the appendix, shows the onset of a remarkable bifurcation phenomenon \AM{changing} the nature of the eigenvalues of $\mathbf{Q}(k)$ when varying the coupling strength $k$ beyond a critical value $k_c$. 
\begin{proposition} 
\label{prop: eigenvalues}
Suppose that $\gamma_i^2-4\omega_i^2>0$ for $i=1,2$. Then for sufficiently small $k$, $\mathbf{Q}(k)$ has four distinct real roots. For sufficiently large $k$, $\mathbf{Q}(k)$ has two distinct real roots and two complex conjugate non-real roots if $8(\omega_1^2+\omega_2^2) < (\gamma_1+\gamma_2)^2$, and has two pairs of non-real complex conjugate roots if $8(\omega_1^2+\omega_2^2) > (\gamma_1+\gamma_2)^2$. In addition, if $\gamma_1=\gamma_2=\gamma$ and $\omega_1=\omega_2=\omega$ (identical beads) then $\mathbf{Q}(k)$ has four distinct real roots for $k<k_c$, has 3 distinct real roots if $k=k_c$ and and has two real roots and a pair of complex conjugate roots for $k>k_c$, where the critical value $k_c$ is given by
$$
k_c=\frac{\gamma^2-4\omega^2}{8}\; .
$$ 
\end{proposition}
%A critical value of the coupling parameter may appear in different scenarios. However, 
In general, finding an explicit formula for $k_c$ is rather nontrivial, for this quantity may depend in a highly nonlinear fashion on the parameters of the model. For instance, the case illustrated in Example \ref{ex: phase transition} details a concrete illustration. \AM{In} Figure \ref{fig:fig1}\AM{, we  display} the behavior of the eigenvalues of $\mathbf{Q}$ as functions of $k$, for fixed values of the other parameters as \AM{chosen} in Example \ref{ex: phase transition}. As expected, the plot highlights the presence of a critical value $k_c$ of the coupling parameter, at which two of the eigenvalues merge: for $k>k_c$ two of the eigenvalues become complex-valued (see also Sec. \ref{ex: phase transition}).
Clearly, for $k=0$ the matrix $\mathbf{Q}$ can be partitioned into four $2\times 2$ blocks, viz.:
\begin{equation}
    \mathbf{Q}=\begin{pmatrix}
\mathbf{Q}_1 & \mathbf{0} \\
\mathbf{0} & \mathbf{Q}_2 
\end{pmatrix} \quad , \quad 
\mathbf{Q}_i=\begin{pmatrix}
0 & 1 \\
-\omega_i^2 & -\gamma_i 
\end{pmatrix} \, , \quad i=1,2 \; , \label{blocks}
\end{equation}
where the block $\mathbf{Q}_i$ refers to the $i$-th oscillator model.
Correspondingly, the eigenvalues of the block $\mathbf{Q}_i$ read:
\begin{equation}
    \lambda_{i}^{\pm}=-\frac{\gamma_i\pm\sqrt{\gamma_i^2-4 \omega_i^2}}{2} \; . \label{eigen0}
\end{equation}
The \AM{hypothesis} $\gamma_i^2-4 \omega_i^2>0$ \AM{of} Proposition \ref{prop: eigenvalues} guarantees that \AM{all} the eigenvalues of $\mathbf{Q}_1$ and $\mathbf{Q}_2$ are real. This is known as the \textit{overdamped regime} \cite{Risken}. Our \AM{subsequent} discussion will thus be restricted to the derivation of reduced descriptions of the original dynamics, given in Eq. \eqref{original}, in such overdamped regime.

\begin{figure}
     \centering
         \includegraphics[width=0.6\textwidth]{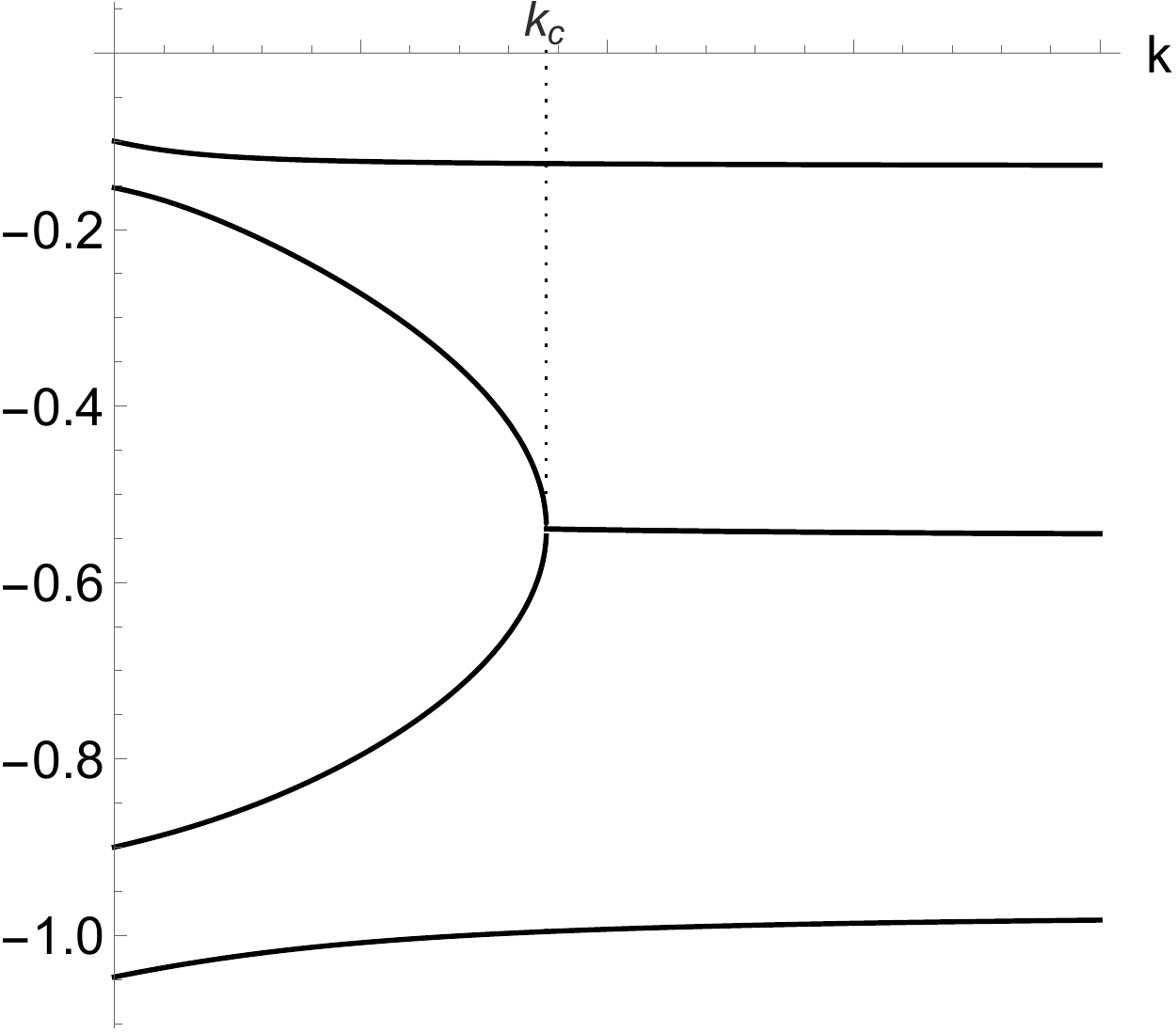}
        \caption{Behavior of the real part of the eigenvalues of $\mathbf{Q}(k)$ in Eq. \eqref{Q} as functions of $k$, with $\omega_1=0.3$, $\omega_2=0.4$, $\gamma_1=1.0$ and $\gamma_2=1.2$. The critical coupling parameter is $k_c\simeq 0.0876$.}
        \label{fig:fig1}
\end{figure}

Exploiting the linearity of the system \eqref{original}, we seek for a closure of the form:
\begin{subequations}
\label{closx2v2}
\begin{align}
   \langle x_2 \rangle &= a \langle x_1 \rangle+ b \langle v_1 \rangle  \label{clos1} \\
    \langle v_2 \rangle &= c \langle x_1 \rangle+ d \langle v_1 \rangle \; ,  \label{clos2}   
\end{align}
\end{subequations}
expressed in terms of the real-valued coefficients $a,b,c,d$.
Using the method of the invariant manifold, cf. Refs. \cite{colan07b,colan09,book_chapter,ColMun22,GorKar05,colan08}, one may compute the time derivative of the observable $\langle x_2 \rangle$, as expressed by the closure \eqref{clos1}, in two different ways.
On the one hand, starting from the definition of the original dynamics, we may write:
\begin{equation}
    \partial_t^{micro} \langle x_2 \rangle:= \langle v_2 \rangle= c \langle x_1 \rangle+ d \langle v_1 \rangle \label{micro1}
\end{equation}
where the second equality directly stems from the closure \eqref{clos2}.
On the other hand, we may apply the chain rule directly in Eq. \eqref{clos1}, which yields:
\begin{align}
    \partial_t^{macro} \langle x_2 \rangle&:= a \langle \dot{x}_1 \rangle+ b \langle \dot{v}_1 \rangle \nonumber\\
%    &= a \langle v_1 \rangle + b \left[-\w_{0,1}^2 \langle x_1\rangle +k \langle x_2\rangle -\gamma_1 \langle v_1 \rangle \right] \nonumber \\
%    &= a \langle v_1 \rangle + b \left[-\w_{0,1}^2 \langle x_1\rangle +k \left( a\langle x_1\rangle + b \langle v_1 \rangle \right) -\gamma_1 \langle v_1 \rangle \right] \nonumber \\
    &= \left(-\w_{0,1}^2~b + k~a~b   \right) \langle x_1 \rangle + \left(a - \gamma_1~b + k~b^2  \right) \langle v_1 \rangle \; . \label{macro1}
\end{align}
One may proceed analogously with the computation of the time derivative of $\langle v_2 \rangle$, thus obtaining:
\begin{align}
    \partial_t^{micro} \langle v_2 \rangle&:= 
    -\w_{0,2}^2 \langle x_2\rangle +k \langle x_1\rangle -\gamma_2 \langle v_2\rangle \nonumber\\
%    &=-\w_{0,2}^2 \left(a \langle x_1\rangle+ b \langle v_1\rangle\right) +k \langle x_1\rangle -\gamma_2 \left(c \langle x_1\rangle+ d \langle v_1\rangle\right) \nonumber\\
    &= \left(-\w_{0,2}^2~a + k - \gamma_2~c\right) \langle x_1\rangle+ \left(-\w_{0,2}^2~b-\gamma_2~d \right)\langle v_1\rangle \label{micro2}
\end{align}
and 
\begin{align}
    \partial_t^{macro} \langle v_2 \rangle&:= c \langle \dot{x}_1 \rangle+ d \langle \dot{v}_1 \rangle \nonumber\\
%    &= c \langle v_1 \rangle + d \left[-\w_{0,1}^2 \langle x_1\rangle +k \langle x_2\rangle -\gamma_1 \langle v_1 \rangle \right] \nonumber \\
%    &= c \langle v_1 \rangle + d \left[-\w_{0,1}^2 \langle x_1\rangle +k \left( a\langle x_1\rangle + b \langle v_1 \rangle \right) -\gamma_1 \langle v_1 \rangle \right] \nonumber \\
    &= \left(-\w_{0,1}^2~d + k~a~d   \right) \langle x_1 \rangle + \left(c - \gamma_1~d + k~b~d  \right) \langle v_1 \rangle \; . \label{macro2}
\end{align}
%\MC{\begin{remark}
Note that Eqs. \eqref{macro1} and \eqref{macro2} yield the projection of the vector field describing the original dynamics of the variables $\langle x_2 \rangle$ and $\langle v_2\rangle$ onto the tangent space to the manifold parameterized by $\langle x_1 \rangle$ and $\langle v_1\rangle$.

%\end{remark}}
The Dynamic Invariance principle \cite{GorKar05} states that the ``microscopic'' and ``macroscopic'' time derivatives of both $\langle x_2 \rangle$ and $\langle v_2 \rangle$ coincide, independently of the values of the observables $\langle x_1 \rangle$ and $\langle v_1 \rangle$. 
Hence, by \AM{equating} \eqref{micro1} with \eqref{macro1}, and \eqref{micro2} with \eqref{macro2}, one obtains a set of nonlinear algebraic equations known as \textit{Invariance Equations} for the unknown coefficients $a,b,c,d$, \AM{viz.}:
\begin{subequations}
\label{inveq}
\begin{align}
 c + \left(\w_{1}^2+k\right)b - k~a~b  &= 0 \; , \label{IE1}\\
 a - \gamma_1~b + k~b^2 -d &= 0 \; ,\label{IE2}\\
 \left(\w_{2}^2+k\right)a - k + \gamma_2~c -\left(\w_{1}^2+k\right)d + k~a~d &= 0 \; ,\label{IE3} \\
\left(\w_{2}^2+k\right)b+\gamma_2~d + c - \gamma_1~d + k~b~d &= 0 \; . \label{IE4}
\end{align}
\end{subequations}
%\textcolor{red}{In general, one does not expect that the invariance equation to have a unique solution. We need some criteria to select a solution.}
The reduced deterministic dynamics can \AM{therefore} be cast in the form:
%\begin{align}
%    \langle \dot{x}_1 \rangle &= \langle v_1 \rangle \nonumber\\
%    \langle\dot{v}_1 \rangle &= - \Omega_{1}(k) \langle x_1 \rangle - \Gamma_1(k) \langle v_1 \rangle \label{red}
%\end{align}
\begin{equation}
    \dot{\mathbf{u}}=\mathbf{M} \mathbf{u}  \, , \label{red}
\end{equation}
with $$\mathbf{u}=(\langle x_1 \rangle,\langle v_1 \rangle)$$ and
\begin{equation}
    \mathbf{M}=-\begin{pmatrix}
0 & -1 \\
\Omega_1^2(k) & \Gamma_1(k)
\end{pmatrix} \, . \label{M}
\end{equation}
Note that the structure of the matrix $\mathbf{M}$ is purposely reminiscent of the matrix $\mathbf{Q}_1$ in Eq. \eqref{blocks}. In fact, the aim of the reduction scheme is to rewrite the original coupled dynamics, as expressed by Eq. \eqref{original2}, in terms of the single-oscillator model equipped with appropriate \textit{generalized coefficients} \cite{Boon}, here denoted by $\Omega_{1}(k)$ and $\Gamma_{1}(k)$. Thus, the variables describing the second oscillator, albeit formally erased, are still properly encoded in the reduced description through the $k$-dependence of the generalized coefficients.
%The latter allow, in fact, to extend $\omega_1$ and $\gamma_1$, which describe the single-oscillator model ($k=0$), to the finite coupling regime.
The latter are defined as:
\begin{subequations}
\label{coeff}
\begin{align}
 \Omega_{1}^2(k)&:= \omega_{1}^2+k(1 - a(k)) \; , \label{Om} \\
 \Gamma_{1}(k)&:= \gamma_{1} - k~b(k) \, ,\label{Ga} 
\end{align}
\end{subequations}
and depend on the coefficients $a, b$ solving the Invariance Equations \eqref{inveq}.
%Note that $\Omega_1(k)$ and $\Gamma_1(k)$ play here the role of \textit{generalized transport coefficients} \cite{Boon}, because they extend the coefficients $\omega_1^2$ and $\gamma_1$, pertaining to the single-oscillator model, to finite values of the coupling parameter $k$.

%\section{Solving the Invariance Equations}
%\label{sec:CE}
%In this Section we explain the use of the Chapman-Enskog method to solve the Invariance Equations as well as the criterion used in determining the right solution of the Invariance Equations.

\subsection{The weak coupling regime}

The Chapman-Enskog (CE) method is a classical tool used in kinetic theory of gases to extract the slow hydrodynamic manifold from the Boltzmann equation; \AM{see, for instance, the works \cite{colan07b,GorKar05}}. When adapted to the present context, the CE method casts into a recurrence procedure yielding approximated solutions of the Invariance Equations \eqref{IE1}-\eqref{IE4}. 

The CE scheme is based on the expansion in powers of a small parameter (e.g. the Knudsen number in Boltzmann's theory) of each of the coefficients $a,b,c,d$, when the values of the other parameters of the model are kept fixed. The expansion parameter, here, is identified with the coupling parameter $k$. We thus write:

\begin{equation}
    a=\sum_{j=0}^\infty k^j a_j \quad , \quad b=\sum_{j=0}^\infty
    k^j b_j \quad , \quad c=\sum_{j=0}^\infty
    k^j c_j \quad , \quad
    d=\sum_{j=0}^\infty
    k^j d_j \, . \label{CE}
\end{equation} 

By inserting the expressions \eqref{CE} into the system \eqref{IE1}--\eqref{IE4}, one finds that the %leading order solution, \textit{viz.}:
%\begin{equation}
%    a_0=b_0=c_0=d_0=0 \, , \label{0ord}
%\end{equation}
%which describes a system made of a single oscillator. The first relevant correction comes from the first-order terms, which take the form:
first relevant contribution comes from the first-order terms, and take the form:
\begin{subequations}
\label{wca}
\begin{align}
     a_1&=\frac{ \gamma_1^2 - \gamma_1 \gamma_2 - \omega_{1}^2 + \omega_{2}^2}{P} \; ,\label{1a} \\
     b_1&=\frac{ \gamma_1 - \gamma_2}{P} \; , \label{1b}\\
     c_1&=\frac{ (-\gamma_1 + \gamma_2)\omega_{1}^2}{P} \; , \label{1c}\\
     d_1&=\frac{- \omega_{1}^2 + \omega_{2}^2}{P} \; , \label{1d}
\end{align}
\end{subequations}
with
$$
P = \gamma_2^2 \omega_{1}^2 +  \gamma_1^2 \omega_{2}^2 + \left(\omega_{1}^2 - \
\omega_{2}^2\right)^2 -  \gamma_1  \gamma_2 \left(\omega_{1}^2 + \omega_{2}^2\right) \, .
$$
\AM{Inserting} the expressions \eqref{1a} and \eqref{1b} into \eqref{Om} and \eqref{Ga}, respectively, and by inserting the resulting formulae in \eqref{red}, \AM{we obtain}  a reduced description that corresponds to the \textit{weak coupling regime}. 
For $j\ge 1$, the coefficients $a_{j},b_{j},c_{j},d_{j}$ \AM{are} obtained \AM{in a direct manner by solving} following recurrence procedure:
\begin{subequations}
\begin{align}
    \omega_{1}^2 b_{j+1}+c_{j+1}&=-b_j+\sum_{\ell=0}^j a_{\ell} b_{j-\ell} \; , \label{CE1}\\
    a_{j+1}-\gamma_1 b_{j+1}&=-\sum_{\ell=0}^j b_{\ell} b_{j-\ell} \; ,\label{CE2}\\
    \omega_{2}^2 a_{j+1}+\gamma_2 c_{j+1}-\omega_{1}^2 d_{j+1} &=-a_j+d_j-\sum_{\ell=0}^j a_{\ell} d_{j-\ell} \; , \label{CE3}\\
    \omega_{2}^2 b_{j+1}+c_{j+1}+\left(\gamma_2-\gamma_1  \right)d_{j+1}&=-b_j-\sum_{\ell=0}^j b_{\ell} b_{j-\ell} \; ,\label{CE4}
\end{align}
\end{subequations}
The system \eqref{CE1}--\eqref{CE4} is \AM{endowed with} the initial conditions \eqref{1a}--\eqref{1d}.

\subsection{Exact solutions of the Invariance Equations}
\label{sec:exact}

Let 
\begin{equation}
    \xi^{\pm}(k):=-\frac{\Gamma_1(k)\pm\sqrt{\Gamma_1^2(k)^2-4 \Omega_1^2(k)}}{2}
    %\; i=1,2 
    \label{eigenval}
\end{equation}
%\begin{equation}
% \lim_{k\rightarrow 0}\xi_{1}^{\pm}(k)=\lambda_{1}^{\pm} \, ,  \label{req}
%\end{equation}
%where
%\begin{equation}
%    \xi_{1}^{\pm}(k)=-\frac{\Gamma_1(k)\pm\sqrt{\Gamma_1^2(k)^2-4 \Omega_1^2(k)}}{2}  \label{eigenval}
%\end{equation}
denote the two eigenvalues of the matrix $\mathbf{M}\in\mathbb{R}^{2\times 2}$ in \eqref{red}.
Among the many sets of solutions $\{a(k),b(k),c(k),d(k)\}$ of the Invariance Equations \eqref{inveq}, the relevant ones are continuous functions \AM{satisfying the asymptotic behavior} 
\begin{equation}
%  \lim_{k\rightarrow 0}\xi_{1}(k)=\lambda_{M}\qquad \text{and} \qquad \lim_{k\rightarrow 0}\xi_{2}(k)=\lambda_{m} \;,\label{req}  
\lim_{k\rightarrow 0}\xi^{+}(k)=\lambda_1^{+}\qquad \text{and} \qquad \lim_{k\rightarrow 0}\xi^{-}(k)=\lambda_{1}^{-} \;,\label{req}  
\end{equation}
%where we set
%\begin{eqnarray}
%\lambda_{m}:=\textrm{min}\{\lambda_{1}^{-},\lambda_{1}^{+}\}\equiv\lambda_{1}^{+} \qquad \text{and} \qquad
%\lambda_{M}:=\textrm{max}\{\lambda_{1}^{-},\lambda_{1}^{+}\}\equiv\lambda_{1}^{-} \; ,\label{minmax}
%\end{eqnarray} 
where $\lambda_1^{\pm}$ are the eigenvalues defined in Eq. \eqref{eigen0}.

\begin{figure}[ht]
     \centering
       \includegraphics[width=0.31\textwidth]{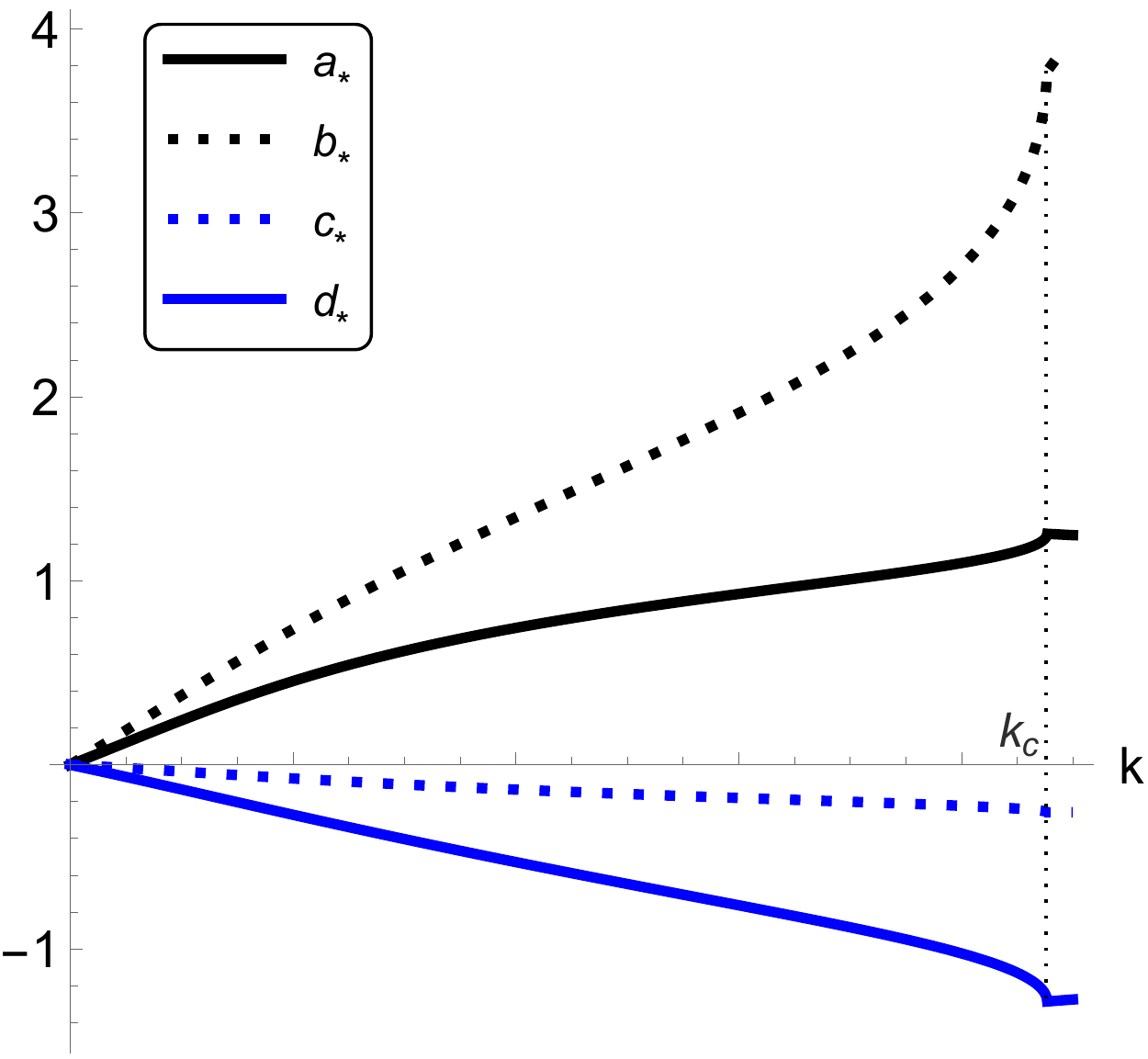}
       \hskip 5pt
       \includegraphics[width=0.3\textwidth]{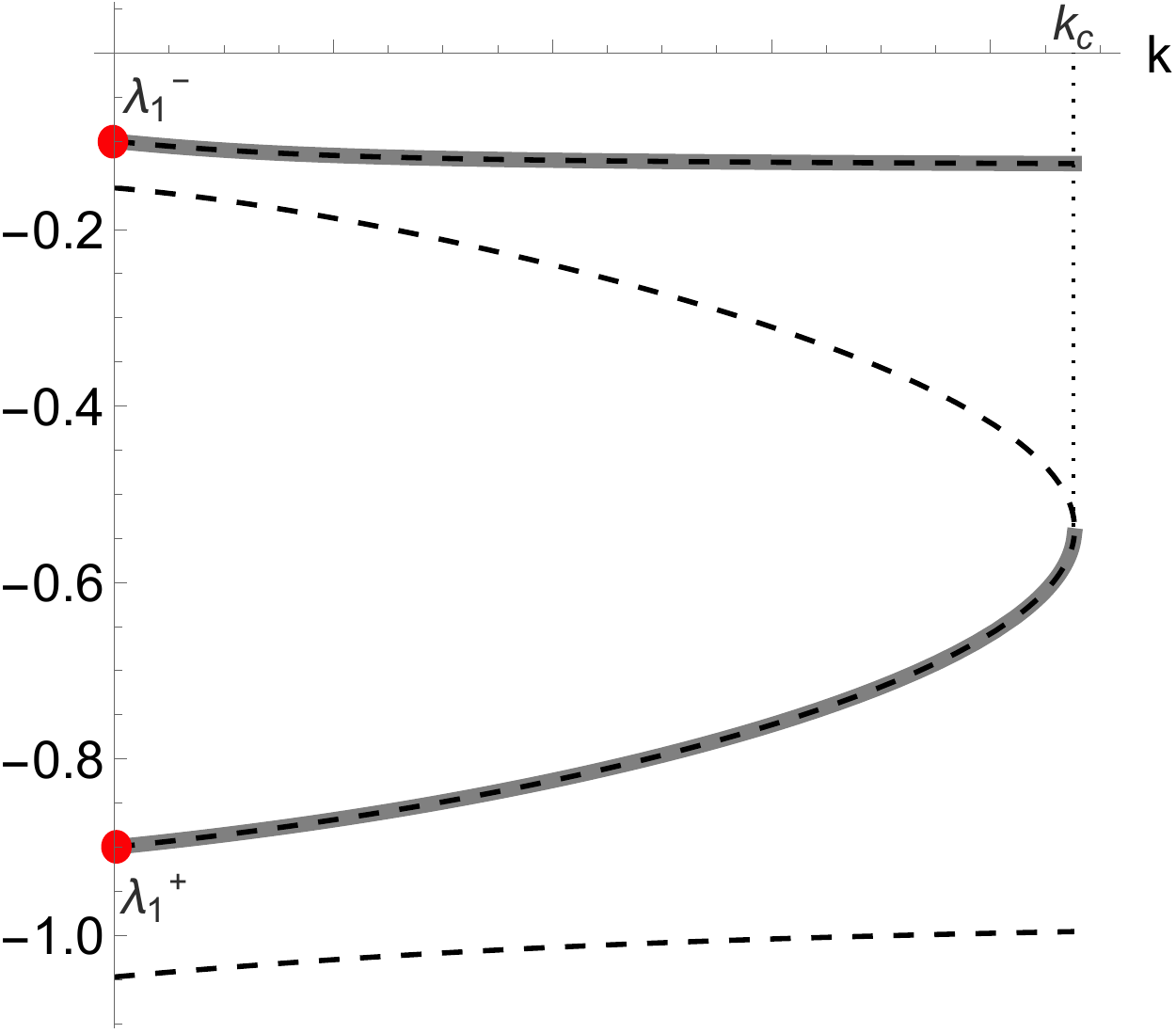}    
       \hskip 5pt
       \includegraphics[width=0.3\textwidth]{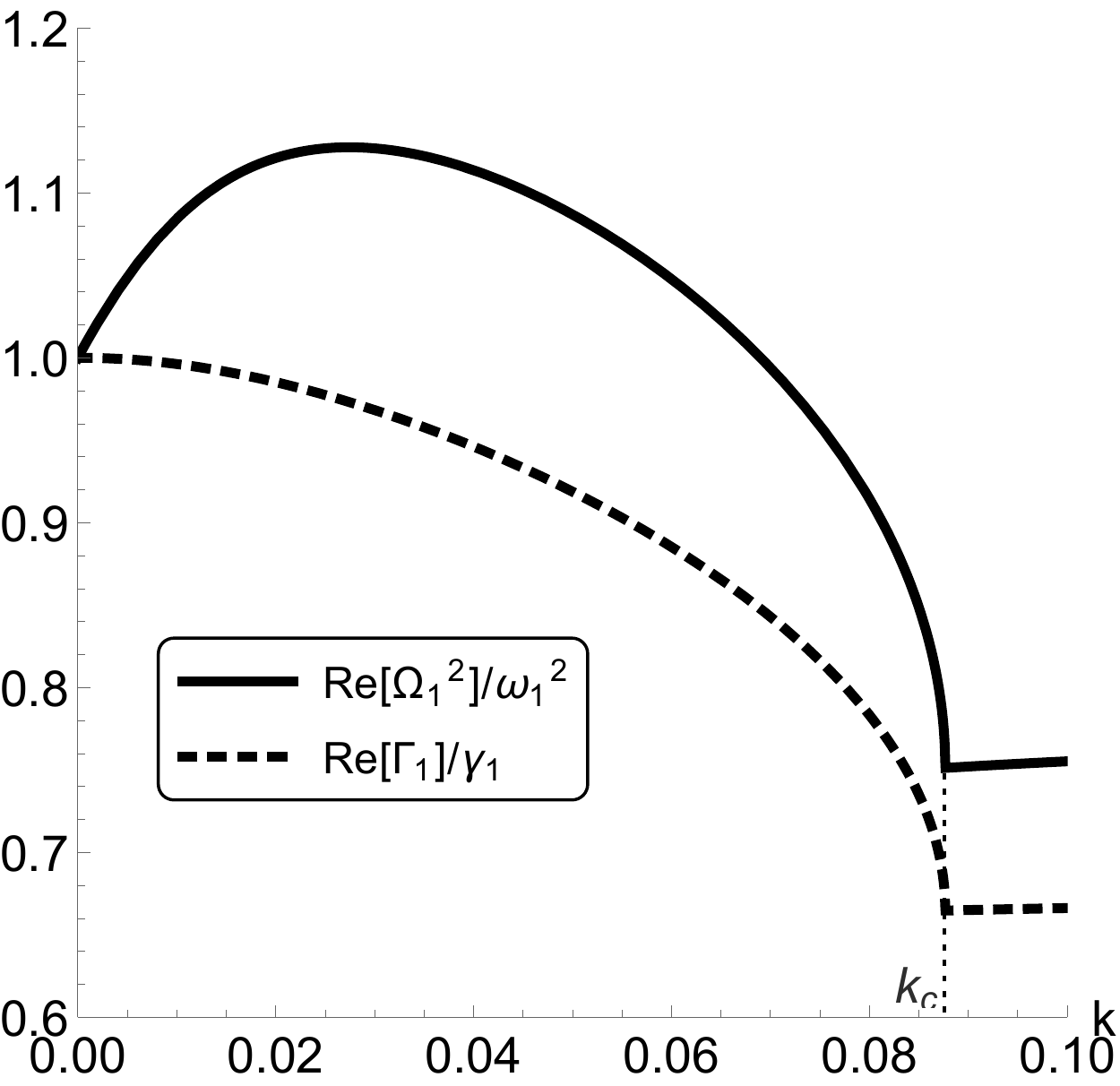}         
        \caption{\textit{Left panel}: Behavior of the real part of the coefficients $\{a_*,b_*,c_*,d_*\}$ solving the system \eqref{IE1}--\eqref{IE4} and verifying \eqref{req}, as functions of $k$. \textit{Central panel}: Behavior of the real part of the eigenvalues $\xi^{\pm}(k)$ (thick gray lines) and the eigenvalues of the original matrix $\mathbf{Q}(k)$ (dashed black lines), as functions of $k$. The red disks on the vertical axis indicate the values of $\lambda_{1}^{\pm}$ defined in \eqref{eigen0}. \textit{Right panel}: Behavior of the real part of $\Omega_1^2$ and $\Gamma_1$, as functions of $k$. In all panels the parameters are set to the same values as in Fig. \ref{fig:fig1}.}
        \label{fig:fig2}
\end{figure}

Our numerical investigation reveals that, for an arbitrary choice of the parameters, the solution with the prescribed asymptotics is unique and is represented by a set $\mathcal{S}_*=\{a_*,b_*,c_*,d_*\}$ shown in the left panel of Fig. \ref{fig:fig2}.
The central panel of Fig. \ref{fig:fig2} provides numerical evidence that the eigenvalues $\xi^{\pm}(k)$ not only reduce to $\lambda_1^{\pm}$ as $k \rightarrow 0$, as dictated by the relation \eqref{req}, but also fully reconstruct, over the interval $k\in[0,k_c]$, the behavior of the corresponding eigenvalues of $\mathbf{Q}(k)$ which attain the same asymptotics. This is clearly evinced, in the central panel, by noticing the perfect overlap between the graphs of the functions $\xi^{\pm}(k)$ and $\lambda_1^{\pm}(k)$, denoted by thick gray lines and dashed black lines, respectively. This is not a mere numerical coincidence: this result stems from evaluating the eigenvalues $\xi^{\pm}(k)$ through the set $\mathcal{S}_*$, which solves the Invariance Equations \eqref{inveq}, cf. e.g. \cite{GorKar05,Kar02}.
The right panel of Fig. \ref{fig:fig2} shows the behavior of the real parts of the generalized coefficients, evaluated with the set $\mathcal{S}_*$. In the interval $[0,k_c]$ the two coefficients $\Omega_1^2(k)$ and $\Gamma_1(k)$ are real-valued, whereas the function $\left(\Gamma_1^2(k)^2-4 \Omega_1^2(k)\right)$ is positive and bounded away from zero. Remarkably, the coefficient $\Omega_1^2$ displays a non-monotonic behavior as a function of the coupling parameter $k$.
We also observe that the definitions \eqref{Om} and \eqref{Ga} imply that the angular frequency and the friction coefficient of the single oscillator are properly recovered as $k\rightarrow 0$ when computing the generalized coefficients with the set $\mathcal{S}_*$, viz.:
\begin{equation}
\lim_{k \rightarrow 0} \Omega_1^2(k)\Big|_{\mathcal{S}_*}=\omega_1^2 \qquad \text{and} \qquad \lim_{k \rightarrow 0} \Gamma_1(k)\Big|_{\mathcal{S}_*}=\gamma_1  \, . \label{limits}
\end{equation}
%Thus, according to \eqref{limits}, the matrix $\mathbf{M}$ in \eqref{M}, evaluated with $\mathcal{S}_*$, properly reconstructs the single-oscillator matrix $\mathbf{Q}_1$ in the small $k$ limit.

The reduced dynamics \eqref{red} supplied with the closure $\mathcal{S}_*$ is thus called \textit{exact}.
%, meaning that the set $\mathcal{S}_*$ not only solves the Invariance Equations \eqref{inveq}, but also yields the proper asymptotics given in \eqref{req}.
It should also be noticed that the weak coupling approximation, represented by Eqs. \eqref{wca}, does not detect the critical value $k_c$, which also affects the original dynamics.
On the contrary, the exact reduced dynamics extends smoothly only up to $k_c$: beyond that value, the functions $\Omega_1^2(k)$ and $\Gamma_1(k)$ become complex-valued. The emergence of a bifurcation phenomenon in the considered model of coupled oscillators, signaled by the presence of a critical value of the coupling parameter, is reminiscent of the linearized Grad's moment system studied in Refs. \cite{colan07,colan07b}.

\subsection{The Fluctuation-Dissipation Theorem}
\label{sec:fdt}

We can now restore noise in the contracted description, under the assumption that the Markovian structure of the original dynamics is preserved.
%We target a suitable Markovian dynamics of the degrees of freedom of the particle $1$, which is the one being retained in the description. This typically stands as an approximation of the exact structure of the reduced dynamics, which is expected to convey some colored noise. Nonetheless, the invoked approximation becomes justified when a vast time scale separation exists between the dynamics of the fast degrees of freedom (those which are removed from the description) and the slow ones.
To this aim, starting from Eq. \eqref{red} we construct a stochastic process for the variables $\{y_1,w_1\}$ (in place of $\{x_1,v_1\}$, respectively) which is cast in the standard Langevin-like form:
\begin{subequations}
\label{red2}
\begin{align}
 \dot{y}_1&=w_1+\sigma_{11}^{(r)}\;\dot{W}_1+\sigma_{12}^{(r)}\;\dot{W}_2 \, , \label{red2a} 
\\
\dot{w}_1&=-\Omega_{1}^2(k) y_1-\Gamma_1(k) w_1+\sigma_{21}^{(r)}\;\dot{W}_1+\sigma_{22}^{(r)}\;\dot{W}_2 \, , \label{red2b}   
\end{align}
\end{subequations}
where the coefficients $\Omega_1^2(k),\Gamma_1(k)$ are given in \eqref{Om} and \eqref{Ga}, and the $\sigma_{ij}^{(r)}$'s, $i,j\in\{1,2\}$, are the elements of the matrix $\boldsymbol{{\sigma}}^{(r)}$ yielding the strength of the noise in the reduced picture.
%\begin{equation}
%    \boldsymbol{{\sigma}}^{(r)}:=\left(
%    \begin{matrix}
%\sigma_{11}^{(r)} & \sigma_{12}^{(r)} 5\\
%\sigma_{21}^{(r)} & \sigma_{22} ^{(r)}
%\end{matrix}
%\right) \label{sigma}\; ,
%\end{equation}
%yield the strength of the noise in the reduced dynamics. 
We then let \begin{equation}
\mathbf{\Sigma}=\frac{1}{2}\boldsymbol{{\sigma}}^{(r)} \left(\boldsymbol{{\sigma}}^{(r)}\right)^T \label{sigma}    
\end{equation} 
denote the diffusion matrix and 
%$\mathbf{R}(t,s)\in \mathbb{R}^{2\times 2}$ , defined as:
\begin{equation}
\mathbf{R}(s,t):=
\left(
    \begin{matrix}
%\langle (y_1(s)-y_1(0))(y_1(t)-y_1(0) \rangle & \langle (y_1(s)-y_1(0))(w_1(t)-w_1(0)) \rangle \\
%\langle (w_1(s)-w_1(0))(y_1(t)-y_1(0)) \rangle & \langle(w_1(s)-w_1(0))(w_1(t)-w_1(0))\rangle
\langle \delta y_1(s) \delta y_1(t)\rangle &
\langle \delta y_1(s) \delta w_1(t) \rangle \\
\langle \delta w_1(s) \delta y_1(t) \rangle & \langle \delta w_1(s) \delta w_1(t) \rangle
\end{matrix}
\right)
\label{covar} \;  
\end{equation}
the covariance matrix, with $\delta y_1(s)\equiv y_1(s)-\langle y_1(s)\rangle$ and $\delta w_1(s)\equiv w_1(s)-\langle w_1(s)\rangle$.
In the long-time limit, the elements of the stationary covariance matrix $\overline{\mathbf{R}}$ are defined as:
\begin{equation}
  \overline{R}_{ij}:=\lim_{t\rightarrow \infty} R_{i,j}(t,t)\quad , \quad i=1,2 \; . \label{steady}  
\end{equation}
The following relation, which is an instance of the Fluctuation-Dissipation Theorem \cite{Kubo,Vulp,Pavl,Zwanzig}, connects the diffusion matrix $\mathbf{\Sigma}$ to the transport matrix $\mathbf{M}$ defined in \eqref{M} and to the stationary covariance matrix $\overline{\mathbf{R}}$:
\begin{equation}
 \mathbf{M}\overline{\mathbf{R}}+\overline{\mathbf{R}} \mathbf{M}^T=-2\mathbf{\Sigma} \label{lyap} \; .
\end{equation}
We remark that while $\mathbf{M}$ has been determined earlier via the invariant manifold method, $\mathbf{\Sigma}$ and $\overline{\mathbf{R}}$ are yet unknown. To overcome this difficulty, we impose that the elements of the matrix $\overline{\mathbf{R}}$ coincide with those of the corresponding matrix evaluated through the original dynamics in Eq. \eqref{original}. Namely, we establish:
\begin{subequations}
\label{equiv}
\begin{align}
   \overline{R}_{11}&= \lim_{t\rightarrow\infty}\langle %(x_1(t)-x_1(0))^2\rangle  \;, \label{equivA}\\
   (\delta x_1(t))^2\rangle  \;, \label{equivA}\\
   \overline{R}_{12}=\overline{R}_{21}&= \lim_{t\rightarrow\infty}\langle 
%   (x_1(t)-x_1(0))( v_1(t)-v_1(0))\rangle \;,\label{equivB} \\
\delta x_1(t) ~ \delta v_1(t)\rangle \;,\label{equivB} \\
   \overline{R}_{22}&= \lim_{t\rightarrow\infty}\langle %(v_1(t)-v_1(0))^2\rangle    \; \label{equivC}.
   (\delta v_1(t))^2\rangle    \;  ,\label{equivC}
\end{align}
\end{subequations}
where, again, we denoted $\delta x_1(s)\equiv x_1(s)-\langle x_1(s)\rangle$ and $\delta v_1(s)\equiv v_1(s)-\langle v_1(s)\rangle$.
Eqs. \eqref{equiv} thus permit to evaluate the diffusion matrix $\mathbf{\Sigma}$ via Eq. \eqref{lyap}.

\begin{remark}
In the analysis of linear stochastic differential equations \cite{Pavl,Risken}, the diffusion matrix is typically assigned, and the Lyapunov equation \eqref{lyap} can hence be used to compute the stationary covariance matrix. In our set-up, the perspective is somehow reversed.  The scale invariance of the stationary covariance matrix permits to fix the matrix $\overline{\mathbf{R}}$ in Eq. \eqref{lyap}, which is then solved to evaluate the diffusion matrix $\mathbf{\Sigma}$.
\end{remark}

%Thus, while the matrix $\mathbf{M}$ is characterized, in a preliminary step, via the Invariant Manifold method and the matrix $\overline{\mathbf{R}}$ is pinned to the original dynamics, the matrix $\mathbf{\Sigma}$ can be determined through Eq. \eqref{lyap}. 
An explicit representation of $\mathbf{\Sigma}$ can be obtained in the case of identical beads, i.e. for $\gamma_1=\gamma_2=\gamma$, $\omega_1=\omega_2=\omega$ and $\sigma_1=\sigma_2=\sigma=\sqrt{2\gamma/(\beta m)}$. To address such case, we denote the eigenvalues defined in Eq. \eqref{eigen0} just as $\lambda^{\pm}$. For later use, we also set:
\begin{equation}
\phi^{\pm}=-\frac{\gamma\pm\sqrt{\gamma^2-4(\omega^2+2k)}}{2} \; \label{e}.    
\end{equation}
Following Chandrasekhar \cite{chandra}, we rewrite Eq. \eqref{original} in the form:
\begin{equation}
    \ddot{x}_i+\gamma \dot{x}_i +\omega^2 x_i+k x_i-k x_j=\sigma \dot{W}_i \; , \; i,j=1,2 \; , i\neq j \; ,\label{orig2}
\end{equation}
and then set $u=x_1+x_2$ and $z=x_1-x_2$. We thus rewrite Eq. \eqref{orig2} in the form
\begin{subequations}
\label{eq}
\begin{align}
\ddot{u}+\gamma \dot{u}+\omega^2 u&=\sigma\dot{\xi}_1\;, \label{equ} \\
\ddot{z}+\gamma \dot{z}+(\omega^2+2 k) z&=\sigma\dot{\xi}_2 \;, \label{eqz} 
\end{align}
\end{subequations}
with $\dot{\xi}_1=\dot{W}_1+\dot{W}_2$, $\dot{\xi}_2=\dot{W}_1-\dot{W}_2$. The two equations \eqref{equ}-\eqref{eqz} are now decoupled and
%We then find the solution of the homogeneous (deterministic) system of ODEs associated to the two Eqs. \eqref{equ}-\eqref{eqz} in the form:
%\begin{eqnarray}
% u(t)&=& A_0 e^{-\lambda_+ t} +B_0 e^{-\lambda_- t} \nonumber \\
% z(t) &=& C_0 e^{-\phi_+ t} +D_0 e^{-\phi_- t} \label{uw} \; ,
%\end{eqnarray}
%where the constants $A_0,\dots,D_0$ are fixed by the initial conditions.
we may express their solutions in the form:
\begin{subequations}
\label{yz}
\begin{align}
%u(t)&=&\mathcal{A}_- e^{-\lambda_- t} + \mathcal{A}_+ e^{-\lambda_+ t} \nonumber \\
% z(t) &=& \mathcal{B}_- e^{-\phi_- t} + \mathcal{B}_+ e^{-\phi_+ t} \; , \label{yz}
u(t)&=\mathcal{A}_- e^{\lambda^- t} + \mathcal{A}_+ e^{\lambda^+ t}\; , \label{yzA} \\
z(t)&= \mathcal{B}_- e^{\phi^- t} + \mathcal{B}_+ e^{\phi^+ t} \; , \label{yzB}
\end{align}
\end{subequations}
%with:
%\begin{eqnarray}
%\lambda_{\pm}&=&\frac{\gamma\pm\sqrt{\gamma^2-4\omega^2}}{2} \label{lamdba} \; ,\nonumber\\
%\phi_{\pm}&=&\frac{\gamma\pm\sqrt{\gamma^2-4(\omega^2+2k)}}{2} \label{phi} \; ,
%\end{eqnarray}
where $\mathcal{A}_{\pm},\mathcal{B}_{\pm}$ are functions of time which can be determined using the method of the variation of the parameters. 
Thus, upon denoting $\dot{\mathcal{O}}\equiv \frac{d}{dt}\mathcal{O}$, we have that:
\begin{subequations}
\label{boh}
\begin{align}
\dot{\mathcal{A}}_-e^{\lambda^- t}+\dot{\mathcal{A}}_+e^{\lambda^+ t}&=0 \; , \label{bohA}\\
\lambda^-\dot{\mathcal{A}}_-e^{\lambda^- t}+\lambda^+\dot{\mathcal{A}}_+e^{\lambda^+ t}&= \sigma \dot{\xi}_1 \; \label{bohB} ,
\end{align}
\end{subequations}
and an analogous set of equations holds for the functions $\mathcal{B}_{\pm}$.
After solving the two sets of coupled ODEs, one arrives at the following expressions of the functions $\mathcal{A}_{\pm}$ and $\mathcal{B}_{\pm}$:
\begin{eqnarray}
\mathcal{A}_{\pm}(t)&=&A_{\pm}(t)+a_{\pm} \; , \nonumber\\
\mathcal{B}_{\pm}(t)&=&B_{\pm}(t)+b_{\pm} \; , \nonumber
\end{eqnarray}
where the constants $a_{\pm},b_{\pm}$ are fixed by the initial conditions, and with:
\begin{subequations}
\label{ApmA}
\begin{align}
%A_{\pm}(t)&=&\pm \frac{\sigma}{(\lambda_+-\lambda_-)}\int_0^t e^{\lambda_{\pm} s} \dot{\xi}_1(s) ds \; , \label{Apm}\\
%B_{\pm}(t)&=&\pm \frac{\sigma}{(\phi_+-\phi_-)}\int_0^t e^{\phi_{\pm} s} \dot{\xi}_2(s) ds \label{Bpm} \;,
A_{\pm}(t)&=\mp \frac{\sigma}{(\lambda^- - \lambda^+)}\int_0^t e^{-\lambda^{\pm} s} \dot{\xi}_1(s) ds \; , \label{Apm}\\
B_{\pm}(t)&=\mp \frac{\sigma}{(\phi^- - \phi^+)}\int_0^t e^{-\phi^{\pm} s} \dot{\xi}_2(s) ds \label{Bpm} \;.
\end{align}
\end{subequations}
We can then provide the analytical solution to the system \eqref{orig2} in the form:
\begin{subequations}
\label{xv1}
\begin{align}
%x_1(t)&=& \frac{1}{2}\left(A_- e^{-\lambda_- t} + A_+ e^{-\lambda_+ t}+ B_- e^{-\phi_- t} + B_+ e^{-\phi_+ t}+\right. \nonumber\\
%&& \left. a_- e^{-\lambda_- t}+a_+ e^{-\lambda_+ t}+b_- e^{-\phi_- t}+b_+ e^{-\phi_+ t}\right) \; , \label{x1b}  \\
%v_1(t)&=& -\frac{1}{2}\left(\lambda_- A_- e^{-\lambda_- t} + \lambda_+ A_+ e^{-\lambda_+ t} + \phi_- B_- e^{-\phi_- t} + \phi_+ B_+ e^{-\phi_+ t}\right. \nonumber\\
%&+& \left. \lambda_- a_- e^{-\lambda_- t}+\lambda_+ a_+ e^{-\lambda_+ t}+ \phi_- b_- e^{-\phi_- t}+b_-\phi_+ e^{-\phi_+ t}\right)  \; . \label{v1b}
x_1(t)&= \frac{1}{2}\left(A_- e^{\lambda^- t} + A_+ e^{\lambda^+ t}+ B_- e^{\phi^- t} + B_+ e^{\phi^+ t}\right. \nonumber\\
&+ \left. a_- e^{\lambda^- t}+a_+ e^{\lambda^+ t}+b_- e^{\phi^- t}+b_+ e^{\phi^+ t}\right) \; , \label{x1b}  \\
v_1(t)&= \frac{1}{2}\left(\lambda^- A_- e^{\lambda^- t} + \lambda^+ A_+ e^{\lambda^+ t} + \phi^- B_- e^{\phi^- t} + \phi^+ B_+ e^{\phi^+ t}\right. \nonumber\\
&+ \left. \lambda^- a_- e^{\lambda^- t}+\lambda^+ a_+ e^{\lambda^+ t}+ \phi^- b_- e^{\phi^- t}+b_-\phi^+ e^{\phi^+ t}\right)  \; . \label{v1b}
\end{align}
\end{subequations}
%\textcolor{red}{(49)-(50) are a bit confusing since those $A_{\pm}, B_{\pm}$ are not the ones shown up, but only the integral terms. should we use other notations like $\bar{A}_{\pm}, \bar{B}_{\pm}$ or so in (49)-(50)? 
%Also in (50) I think it should be $-\lambda_+ A_+ e^{-\lambda_+ t}$ and $- \phi_+ B_+ e^{-\phi_+ t}$??
%}
%\textcolor{red}{I think (50) is even simpler:
%$$
%v_1(t)=\dot{x}_1(t)=\frac{1}{2}(\dot{u}(t)+\dot{z%}(t))=\frac{1}{2}(-\lambda_- A_ e^{-\lambda_{-} %t}-\lambda_{+} A_{+} e^{-\lambda_{+} t})
%$$
%Also in (50) I think it should be $-\lambda_+ A_+ %e^{-\lambda_+ t}$ and $- \phi_+ B_+ e^{-\phi_+ %t}$??
%}
To proceed further, it is necessary to evaluate the following averages:
\begin{eqnarray*}
%\langle A_{\pm}^2(t)\rangle &=&
%\frac{\sigma^2}{(\lambda_+-\lambda_-)^2} \int_0^t e^{\lambda_{\pm} s} ds \int_0^{t} e^{\lambda_{\pm} \tau} \langle \dot{\xi}_1(s) \dot{\xi}_1(\tau)  \rangle d\tau \\
%&=&\frac{1}{ \lambda_{\pm}}\frac{\sigma^2}{(\lambda_+-\lambda_-)^2} \left(e^{2\lambda_{\pm} t}-1\right)\;, \\
%\langle A_{+}(t)A_{-}(t) \rangle &=&   -\frac{2}{(\lambda_{+}+\lambda_-)}\frac{\sigma^2}{(\lambda_+-\lambda_-)^2} \left(e^{(\lambda_{+}+\lambda_-) t}-1\right)\;, \\
%\langle A_{+}(t)B_{\pm}(t) \rangle &=& \langle A_{-}(t)B_{\pm}(t) \rangle = 0  \;, \\
%\langle B_{\pm}^2(t) \rangle &=& \frac{1}{ \phi_{\pm}}\frac{\sigma^2}{(\phi_+-\phi_-)^2} \left(e^{2\phi_{\pm} t}-1\right) \; , \\
% \langle B_{+}(t)B_{-}(t) \rangle &=& -\frac{2}{(\phi_{+}+\phi_-)}\frac{\sigma^2}{(\phi_+-\phi_-)^2} \left(e^{(\phi_{+}+\phi_-) t}-1\right) \; ,
\langle A_{\pm}^2(t)\rangle &=&
\frac{\sigma^2}{(\lambda^- -\lambda^+)^2} \int_0^t e^{\lambda_{\pm} s} ds \int_0^{t} e^{-\lambda^{\pm} \tau} \langle \dot{\xi}_1(s) \dot{\xi}_1(\tau)  \rangle d\tau \\
&=&-\frac{1}{\lambda^{\pm}}\frac{\sigma^2}{(\lambda^- -\lambda^+)^2} \left(e^{-2\lambda^{\pm} t}-1\right)\;, \\
\langle A_{+}(t)A_{-}(t) \rangle &=&   \frac{2}{(\lambda^{+}+\lambda^-)}\frac{\sigma^2}{(\lambda^- -\lambda^+)^2} \left(e^{-(\lambda^{+}+\lambda^-) t}-1\right)\;, \\
\langle A_{+}(t)B_{\pm}(t) \rangle &=& \langle A_{-}(t)B_{\pm}(t) \rangle = 0  \;, \\
\langle B_{\pm}^2(t) \rangle &=& -\frac{1}{ \phi^{\pm}}\frac{\sigma^2}{(\phi^- -\phi^+)^2} \left(e^{-2\phi^{\pm} t}-1\right) \; , \\
 \langle B_{+}(t)B_{-}(t) \rangle &=& \frac{2}{(\phi^{+}+\phi_-)}\frac{\sigma^2}{(\phi^- -\phi^+)^2} \left(e^{-(\phi^{+}+\phi^-) t}-1\right) \; ,
\end{eqnarray*}
where we used that $\left\langle \dot{W}_1(s)\dot{W}_2(t)\right\rangle=0$, for any $s,t >0$.

After some straightforward, albeit lengthy, calculations one finds:
\begin{subequations}
\label{Rij}
\begin{align}
%\overline{R}_{11}(k)&=&\frac{\sigma^2}{4}\left[\frac{1}{(\lambda_+-\lambda_-)^2}\left(\frac{1}{\lambda_-}-\frac{4}{\lambda_+ +\lambda_-}+\frac{1}{\lambda_+}  \right) \nonumber \right. \\
%&\qquad+& \left.\frac{1}{4(\phi_+-\phi_-)^2}\left(\frac{1}{\phi_-}-\frac{4}{\phi_+ +  \phi_-}+\frac{1}{\phi_+}  \right)
%\right] \; , \label{R11} \\
%\overline{R}_{12}(k)=\overline{R}_{21}(k)&=&0 \;, \label{Rija} \\
%\overline{R}_{22}(k)&=&\frac{\sigma^2}{2 \gamma} \; . \label{R22} 
\overline{R}_{11}(k)&=\frac{\sigma^2}{4}\left[\frac{1}{(\lambda^- -\lambda^+)^2}\left(-\frac{1}{\lambda^-}+\frac{4}{\lambda^+ +\lambda^-}-\frac{1}{\lambda^+}  \right) \nonumber \right. \\
&\qquad+ \left.\frac{1}{(\phi^- -\phi^+)^2}\left(-\frac{1}{\phi^-}+\frac{4}{\phi^+ +  \phi^-}-\frac{1}{\phi_+}  \right)
\right] \; , \label{R11} \\
\overline{R}_{12}(k)=\overline{R}_{21}(k)&=0 \;, \label{Rija} \\
\overline{R}_{22}(k)&=\frac{\sigma^2}{2 \gamma} \; . \label{R22} 
\end{align}
\end{subequations}
We observe that for $k=0$ it holds $\overline{R}_{11}(0)=(\beta m \omega^2)^{-1}$ and $\overline{R}_{22}(0)=(\beta m)^{-1}$. 
%as also shown in Fig. \ref{fig:fig5}.
From Eq. \eqref{lyap}, one finally determines the diffusion matrix $\mathbf{\Sigma}$, whose elements read:
\begin{subequations}
\label{sig}
\begin{align}
\Sigma_{11}(k)&= 
%-  \overline{R}_{12}=
0 \; ,\label{sig11}\\
\Sigma_{12}(k)&=\Sigma_{21}(k)=
%\frac{1}{2}\left(
%-\overline{R}_{22}+\Gamma_1(k) \overline{R}_{12}+\Omega_1^2(k) \overline{R}_{11}\right) \nonumber\\
-\frac{1}{2}\left(\overline{R}_{22}+\Omega_1^2(k) \overline{R}_{11}\right) \; , \label{sig12}\\
\Sigma_{22}(k)&= 
%\Gamma_1(k) \overline{R}_{22} + \Omega_1^2(k) \overline{R}_{12}= 
\Gamma_1(k) \overline{R}_{22} \label{sig22} \; .
\end{align}
\end{subequations}
where $\Omega_1^2(k)$ and $\Gamma_1(k)$ are given in Eqs. \eqref{Om} and \eqref{Ga}.

In the absence of coupling, for $k=0$, one correctly finds $\Sigma_{11}(0)=\Sigma_{12}(0)=\Sigma_{21}(0)=0$ and $\Sigma_{22}(0)= \gamma/(\beta m)=\sigma^2/2$, thus recovering the structure of the diffusion matrix pertaining to a single underdamped oscillator. 
The behavior of the real parts of $\Sigma_{12}$ and $\Sigma_{22}$ as functions of $k$ is illustrated in Fig. \ref{fig:fig3}. The figure clearly shows that the coupling between the two Brownian oscillators not only alters the value of $\Sigma_{22}$ with respect to the model with no coupling, but gives also rise to non-zero off-diagonal entries $\Sigma_{12}=\Sigma_{21}$ of the diffusion matrix.
%The numerical solution of Eqs. \eqref{IE1}-\eqref{IE4} for identical beads, along with the corresponding behavior of the functions $\Omega_1^2(k)$, $\Gamma_1(k)$, $\Sigma_{12}(k)$ and $\Sigma_{22}(k)$ are shown in Fig. \ref{fig:fig3}.

\begin{figure}[ht]
     \centering
%        \includegraphics[width=0.3\textwidth]{fig6a} 
%        \hskip 5pt
%        \includegraphics[width=0.3\textwidth]{fig6b} 
%        \hskip 5pt
        \includegraphics[width=0.6\textwidth]{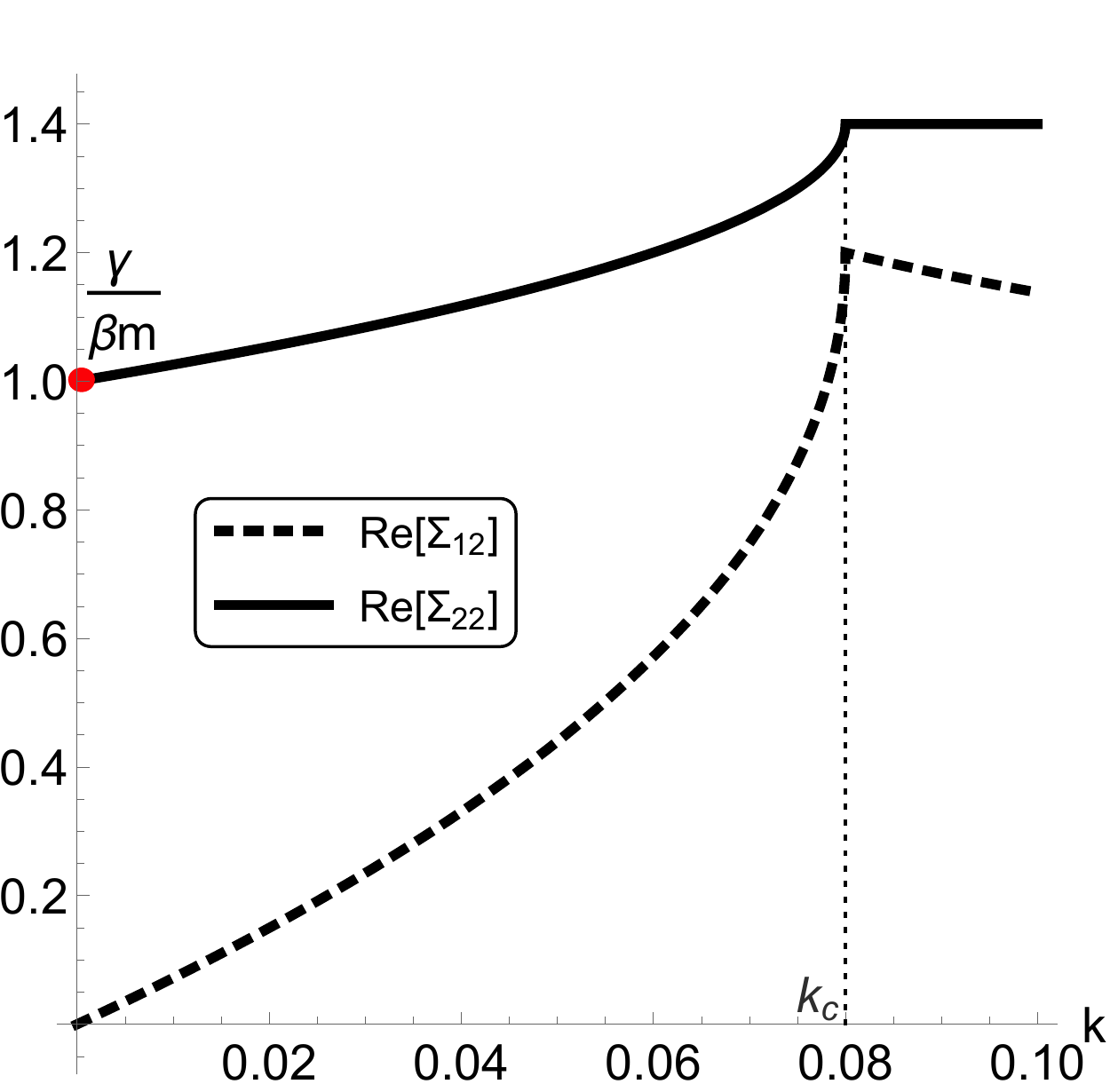}         
        \caption{%Behavior of the real part of $a_*,\dots,d_*$ (left panel), of the real part of $\Omega_1^2(k)$ and $\Gamma_1(k)$ (central panel), and of the real part of $\Sigma_{12}$ and $\Sigma_{22}$ (right panel) as functions of $k$. In the right panel, the red disk at $k=0$ corresponds to the value $\gamma/\beta m$. In all panels, the parameters are set to the values $\omega_1=\omega_2=\omega=0.3$ and $\gamma_1=\gamma_2=\gamma=1$.
        Behavior of the real part of $\Sigma_{12}$ and $\Sigma_{22}$ as functions of $k$ for identical beads for the reduced dynamics defined in Eq. \eqref{red2}. The red disk on the vertical axis, at $k=0$, corresponds to the reference value $\gamma/\beta m$. The parameters are set to the values $\omega_1=\omega_2=\omega=0.3$, $\gamma_1=\gamma_2=\gamma=1$, $\beta=1$ and $m=1$.
        } 
        \label{fig:fig3}
\end{figure}

%, and $\sigma_{1}^{(r)}$ denotes the strength of the noise in the reduced dynamics, related to $\Gamma_1$ and to the inverse temperature $\beta$ via the Fluctuation-Dissipation theorem:
%\begin{equation}
%    \sigma_{1}^{(r)}=\sqrt{\frac{2 \Gamma_1 }{m \beta}} \label{fdt}
%\end{equation}
%Moreover, by letting
%\begin{equation}
%    D_1=\frac{1}{m \beta \gamma_1} \quad \text{and} \quad \mathcal{D}_{1,r}=\frac{1}{m\Gamma_1 \beta} \label{Diff}
%\end{equation}
%denote the diffusion coefficients of the original and the reduced dynamics, respectively, one can determine $\mathcal{D}_r$ through the relation:
%\begin{equation}
%  \mathcal{D}_{1,r} =D_1 \gamma_1 \Gamma_1^{-1} \, . \label{Dr2}  
%\end{equation}
%Thus, the reduced dynamics \eqref{red2} can equivalently be written in the form:
%\begin{align}
% \dot{x}_1&=v_1 \, , \nonumber
%\\\dot{v}_1&=-\Omega_{1}^2 x_1-\Gamma_1 v_1+\sqrt{2 \mathcal{D}_{1,r} \Gamma_1^2 }\; \dot{W}_1 \, . \label{red2}   
%\end{align}

%\begin{remark}
%For identical beads it holds $\Gamma_1=\gamma_1$ and hence $\mathcal{D}_r=D_1$. 
%\end{remark}

\section{Coupled overdamped oscillators}\label{sec:IM2}
We now proceed on a different route along our scheme of model reduction: we aim at eliminating the velocities from the original dynamics \eqref{original} (i.e. the top right arrow in Fig. \ref{fig: diagram}). 
%Later on, we will reduce the description even further simplify the dynamics via two different routes obtaining a reduced dynamics for only the position variable $x_1$ as illustrated in Figure \ref{fig: diagram}. We show the comunity of the diagram.  
The procedure here follows the guidelines previously traced in Sec. \ref{sec:IM}. As long as no ambiguity arises, we shall also retain much of the notation used earlier.
We thus look for a closure of the form:
\begin{subequations}
\begin{align}
   \langle v_1 \rangle &= \overline{a} \langle x_1 \rangle+ \overline{b} \langle x_2 \rangle,   \\
    \langle v_2 \rangle &= \overline{c} \langle x_1 \rangle+ \overline{d} \langle x_2 \rangle \; ,  
\end{align}
\end{subequations}
where the real-valued coefficients $\overline{a},\overline{b},\overline{c},\overline{d}$ must be determined as functions of the parameters of the model. Next, we set:
\begin{equation*}
% \partial_t^{micro} \langle v_1 \rangle :=   -(\w_{0,1}^2+ \gamma_1 ~a)\langle x_1 \rangle+(k- \gamma_1 ~b) \langle x_2 \rangle \; ,
\partial_t^{micro} \langle v_1 \rangle :=   -(\w_{0,1}^2+ \gamma_1 ~\overline{a})\langle x_1 \rangle+(k- \gamma_1 ~\overline{b}) \langle x_2 \rangle \; ,
\end{equation*}
%\begin{align*}
%    \partial_t^{micro} \langle v_1 \rangle&:= %-\w_{0,1}^2\langle x_1 \rangle + k \langle x_2 \rangle -\gamma_1 (a \langle x_1 \rangle+ b \langle x_2 \rangle)\nonumber\\ 
%    &= -(\w_{0,1}^2+ \gamma_1 ~a)\langle x_1 \rangle+(k- \gamma_1 ~b) \langle x_2 \rangle,
%\end{align*}
and
\begin{equation*}
% \partial_t^{macro} \langle v_1 \rangle :=   \left(a^2+b~c\right) \langle x_1 \rangle + \left(a~b+b~d  \right) \langle x_2 \rangle \;,
 \partial_t^{macro} \langle v_1 \rangle :=   \left(\overline{a}^2+\overline{b}~\overline{c}\right) \langle x_1 \rangle + \left(\overline{a}~\overline{b}+\overline{b}~\overline{d}  \right) \langle x_2 \rangle \;,
\end{equation*}
%\begin{align*}
%    \partial_t^{macro} \langle v_1 \rangle&:= 
%    a \langle \dot{x}_1 \rangle+ b \langle \dot{x}_2 \rangle \nonumber\\
%    &= a (a \langle x_1 \rangle+ b \langle x_2\rangle)  + b (c \langle x_1 \rangle+ d \langle x_2\rangle) \nonumber \\
%    &= \left(a^2+b~c\right) \langle x_1 \rangle + \left(a~b+b~d  \right) \langle x_2 \rangle \;, 
%\end{align*}
as well as
\begin{equation*}
% \partial_t^{micro} \langle v_2 \rangle :=   (k-\gamma_2~c)\langle x_1 \rangle-(\w_{0,2}^2+ \gamma_2~d) \langle x_2 \rangle \; ,
 \partial_t^{micro} \langle v_2 \rangle :=   (k-\gamma_2~\overline{c})\langle x_1 \rangle-(\w_{0,2}^2+ \gamma_2~\overline{d}) \langle x_2 \rangle \; ,
\end{equation*}
%\begin{align*}
%    \partial_t^{micro} \langle v_2 \rangle&:= %-\w_{0,2}^2\langle x_2 \rangle + k \langle x_1 \rangle -\gamma_2 (c \langle x_1 \rangle+ d \langle x_2 \rangle)\nonumber\\ 
%    &= (k-\gamma_2~c)\langle x_1 \rangle-(\w_{0,2}^2+ \gamma_2~d) \langle x_2 \rangle,
%\end{align*}
and
\begin{equation*}
% \partial_t^{macro} \langle v_2 \rangle :=   \left(a~c+d~c\right) \langle x_1 \rangle + \left(b~c+d^2  \right) \langle x_2 \rangle \;,
 \partial_t^{macro} \langle v_2 \rangle :=   \left(\overline{a}~\overline{c}+\overline{d}~\overline{c}\right) \langle x_1 \rangle + \left(\overline{b}~\overline{c}+\overline{d}^2  \right) \langle x_2 \rangle \;.
\end{equation*}
%\begin{align*}
%    \partial_t^{macro} \langle v_2 \rangle&:= 
%    c \langle \dot{x}_1 \rangle+ d \langle \dot{x}_2 \rangle \nonumber\\
%    &= c (a \langle x_1 \rangle+ b \langle x_2\rangle)  + d (c \langle x_1 \rangle+ d \langle x_2\rangle) \nonumber \\
%    &= \left(a~c+d~c\right) \langle x_1 \rangle + \left(b~c+d^2  \right) \langle x_2 \rangle.
%\end{align*}
The Invariance Equations can hence be cast in the form:
\begin{subequations}
\label{inveq2}
\begin{align}
% a^2+b~c +\w_{1}^2+k+\gamma_1~a  &= 0, \label{IE1b}\\
% a~b+b~d -k+\gamma_1~b &= 0, \label{IE2b}\\
% a~c+c~d-k+\gamma_2~c  &= 0, \label{IE32} \\
%b~c+d^2+\w_{2}^2+k+\gamma_2~d &= 0. \label{IE4b}
 \overline{a}^2+\overline{b}~\overline{c} +\w_{1}^2+k+\gamma_1~\overline{a}  &= 0, \label{IE1b}\\
 \overline{a}~\overline{b}+\overline{b}~\overline{d} -k+\gamma_1~\overline{b} &= 0, \label{IE2b}\\
 \overline{a}\overline{c}+\overline{c}~\overline{d}-k+\gamma_2~\overline{c}  &= 0, \label{IE32} \\
\overline{b}~\overline{c}+\overline{d}^2+\w_{2}^2+k+\gamma_2~\overline{d} &= 0. \label{IE4b}
\end{align}
\end{subequations}
We then write the reduced deterministic dynamics as:
%\begin{align}
%    \langle \dot{x}_1 \rangle &= \langle v_1 \rangle \nonumber\\
%    \langle\dot{v}_1 \rangle &= - \Omega_{1}(k) \langle x_1 \rangle - \Gamma_1(k) \langle v_1 \rangle \label{red}
%\end{align}
\begin{equation}
    \dot{\mathbf{x}}=\mathbf{P} \mathbf{x}  \, , \label{red2b2}
\end{equation}
with $\mathbf{x}$ denoting the vector of variables $\mathbf{x}=(\langle x_1 \rangle,\langle x_2 \rangle)$, and $\mathbf{P}\in\mathbb{R}^{2 \times 2}$ the matrix 
\begin{equation}
%    \mathbf{P}=\begin{pmatrix}
%a & b \\
%c & d
%\end{pmatrix} \, . \label{M2}
    \mathbf{P}=\begin{pmatrix}
\overline{a} & \overline{b} \\
\overline{c} & \overline{d}
\end{pmatrix} \, . \label{M2}
\end{equation}
Let 
\begin{equation}
%    \zeta_{i}(k)=\frac{(a+d)+(-1)^{i+1}\sqrt{a^2+4 b c-2 a d+d^2}}{2} \; , \; i=1,2 \; , \label{eigenval2}
    \zeta^{\pm}(k)=\frac{(\overline{a}+\overline{d})\pm\sqrt{(\overline{a}-\overline{d})^2+4 \overline{b} \overline{c}}}{2}  \label{eigenval2}
\end{equation}
denote the two distinct eigenvalues of $\mathbf{P}$ and call $\zeta_M(k)$ and $\zeta_m(k)$ the larger and the smaller of them, respectively, in the interval $[0,k_c]$. We also set
\begin{equation}
\lambda_{m}:=\textrm{min}\{\lambda_{1}^{-},\lambda_{2}^{-}\} \qquad \text{and} \qquad
\lambda_{M}:=\textrm{max}\{\lambda_{1}^{-},\lambda_{2}^{-}\} \; ,\label{minmax2}
\end{equation}
where $\lambda_i^-$, $i=1,2$, represents the ``slow'' eigenvalue in the uncoupled dynamics of the $i$-th oscillator.
Thus, among the various sets of solutions $\{\overline{a}(k),\overline{b}(k),\overline{c}(k),\overline{d}(k)\}$ of the Invariance Equations \eqref{inveq2}, the relevant ones, in this context, are continuous functions with asymptotics:
\begin{equation}
  \lim_{k\rightarrow 0}\zeta_{M}(k)=\lambda_{M}\qquad \text{and} \qquad \lim_{k\rightarrow 0}\zeta_{m}(k)=\lambda_{m} \label{req2}  \; .
\end{equation}

\begin{figure}[ht]
     \centering
       \includegraphics[width=0.45\textwidth]{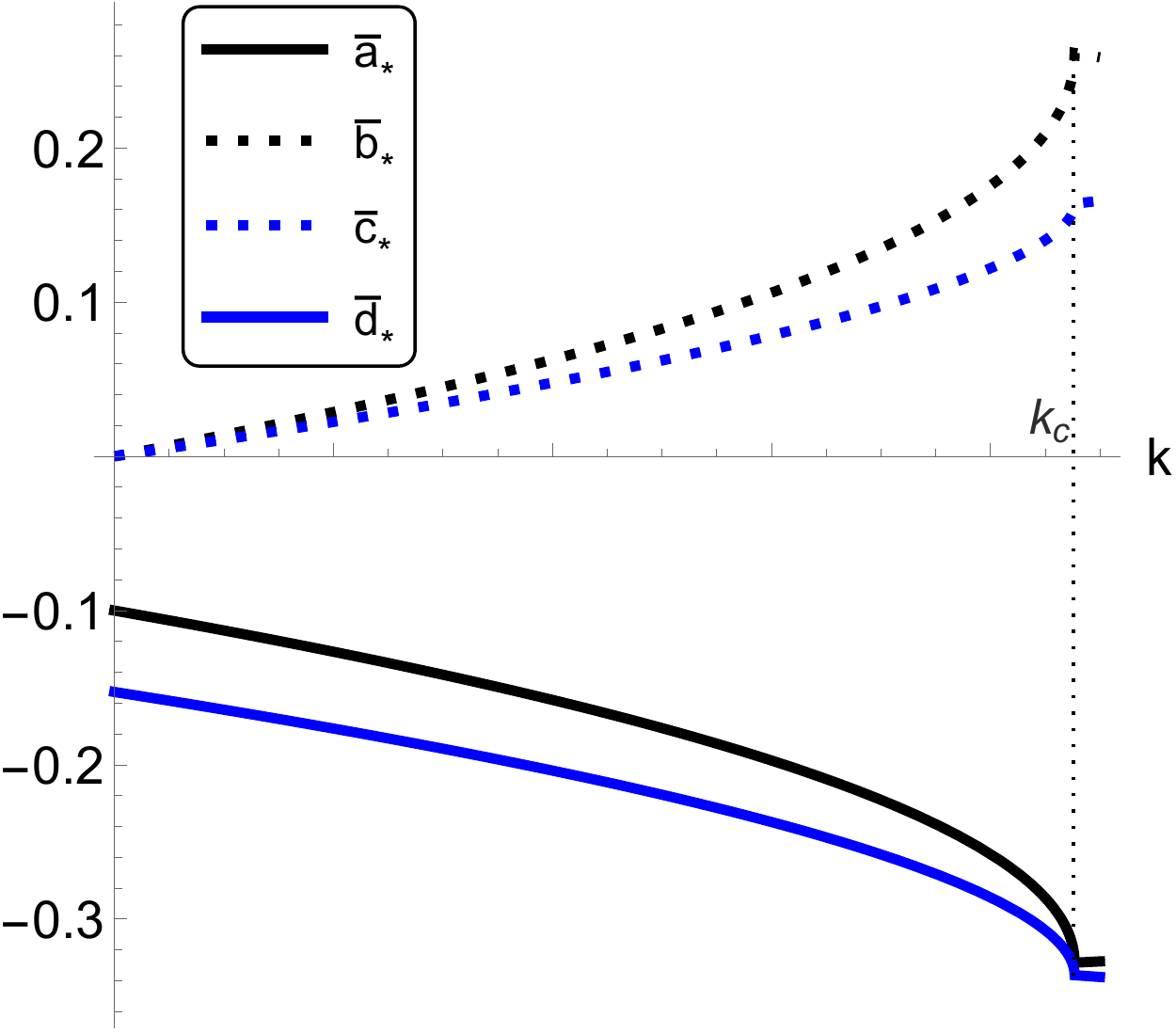}
       \hskip 20pt
       \includegraphics[width=0.41\textwidth]{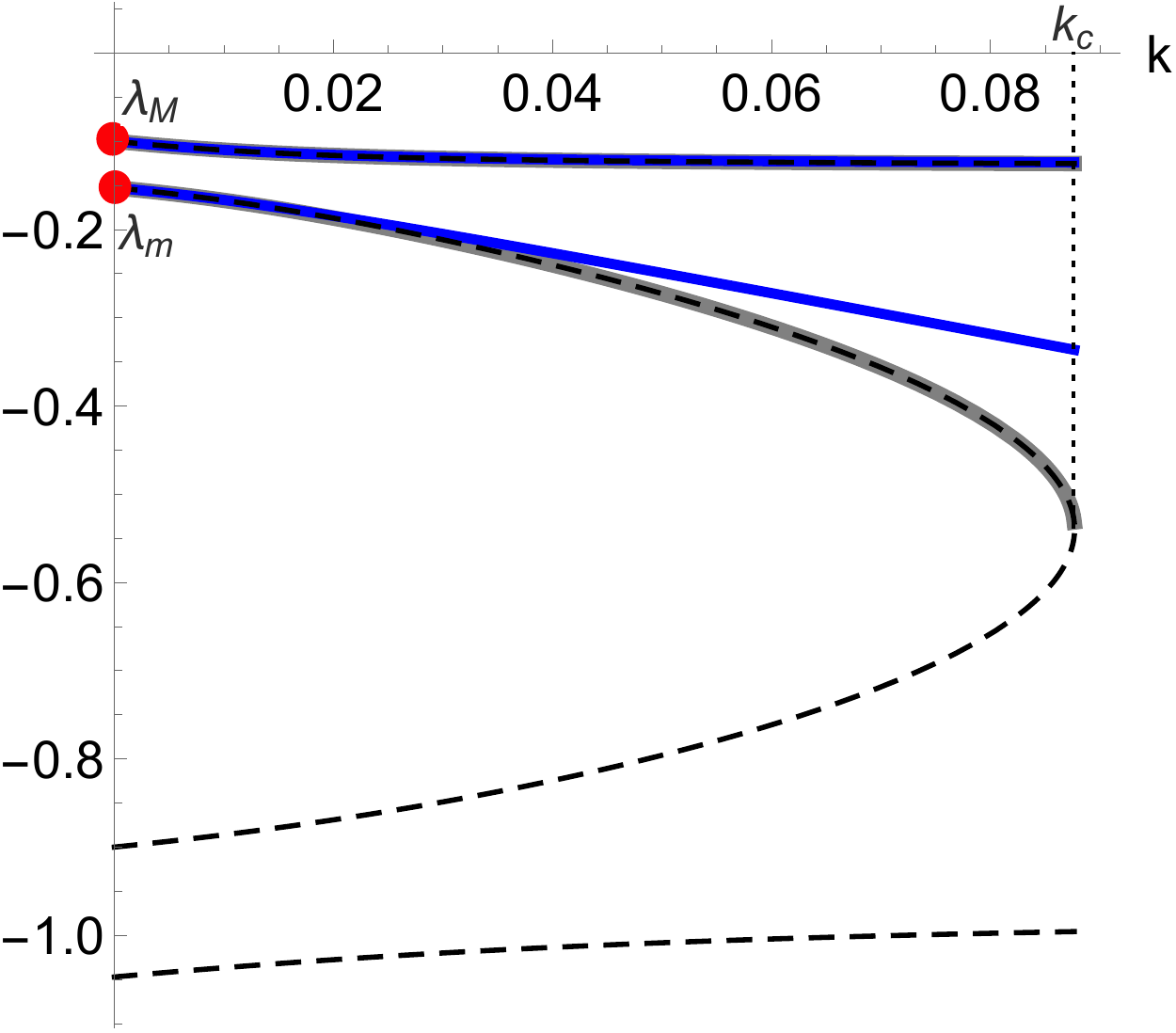}       
        \caption{\textit{Left panel}: Behavior of the real part of the coefficients $\overline{a}_*,\overline{b}_*,\overline{c}_*,\overline{d}_*$ solving the system \eqref{IE1b}--\eqref{IE4b} and verifying \eqref{req2}, as functions of $k$. \textit{Right panel}: Behavior of the real part of the eigenvalues $\zeta_{i}(k)$ (thick gray lines) and the eigenvalues of $\mathbf{Q}$ (dashed black lines), as functions of $k$. The red disks on the vertical axis indicate the values of $\lambda_{m}$ and $\lambda_{M}$ given in Eq. \eqref{minmax2}. The blue lines correspond to the weak coupling regime represented by Eqs. \eqref{a1}-\eqref{d1}. In both panels the parameters are fixed to the same values as in Fig. \ref{fig:fig1}.}
        \label{fig:fig4}
\end{figure}
\AM{In the left panel} of Fig. \ref{fig:fig4} the set of coefficients $\overline{\mathcal{S}}_*=\{\overline{a}_*,\overline{b}_*,\overline{c}_*,\overline{d}_*\}$ solving Eqs. \eqref{inveq2} and fulfilling the asymptotics \eqref{req2} are shown as functions of $k$. 
It is worth noticing that \AM{the off-diagonal entries of the matrix $\mathbf{P}$ vanish} \AM{in the limit} $k\rightarrow 0$, whereas in the same limit the functions $\overline{a}_*$ and $\overline{d}_*$ attain values corresponding to the pair of eigenvalues $\lambda_1^-$ and $\lambda_2^-$.
%, in agreement with Eq. \eqref{req2}.
The right panel of Fig. \ref{fig:fig4} illustrates the behavior of the real part of the eigenvalues $\zeta^{\pm}$ as functions of $k$.  
%This is also confirmed by an inspection of the leading order terms of the corresponding Chapman-Enskog expansion, see Eqs. \eqref{a0}--\eqref{d0} below. 
%\AM{
%It should be noted that the functions $\zeta^{\pm}(k)$ properly recover, in the $k\rightarrow 0$ limit,} the result obtained in Ref. \cite{ColMun22} concerning a single overdamped oscillator model. Even more 
In analogy with the results of Sec. \ref{sec:exact}, the graphs of the two functions $\zeta^{\pm}(k)$ are found to coincide, over the interval $[0,k_c]$, with the graphs of the branches of the spectrum of the matrix $\mathbf{Q}(k)$ sharing the same asymptotics as $k\rightarrow 0$. 
\AM{Furthermore, a straightforward application} of the Chapman-Enskog scheme yields the \AM{precise}  expressions of the coefficients $\{a_*,b_*,c_*,d_*\}$ in the weak coupling regime.
One finds:
\begin{subequations}
\label{boh2}
\begin{align}
\overline{a}_0&=-\frac{\gamma_1 - \sqrt{\gamma_1^2-4 \omega_1^2}}{2}=\lambda_1^{-} \label{a0}\\
\overline{b}_0&=\overline{c}_0=0 \label{bc0} \\
\overline{d}_0&=-\frac{\gamma_2 - \sqrt{\gamma_2^2-4 \omega_2^2}}{2}=\lambda_2^{-} \label{d0}
\end{align}
\end{subequations}
and
\begin{subequations}
\label{wk}
\begin{align}
\overline{a}_1&=-\frac{1}{\sqrt{\gamma_1^2-4\omega_1^2}} \label {a1}\\
\overline{b}_1&=\frac{2}{\gamma_1-\gamma_2+S} \label {b1}\\
\overline{c}_1&=\frac{2}{\gamma_2-\gamma_1+S} \label {c1}\\
\overline{d}_1&=-\frac{1}{\sqrt{\gamma_2^2-4\omega_2^2}} \label {d1}
\end{align}
\end{subequations}

with
\begin{equation}
    S=\sqrt{\gamma_1^2-4\omega_1^2}+ \sqrt{\gamma_2^2-4\omega_2^2}  \; .\label {S}
\end{equation}

At variance with the exact reduced dynamics \eqref{red2b2}, evaluated with the set $\overline{\mathcal{S}}_*$, the weak coupling approximation, denoted by the solid blue lines in the right panel of Fig. \ref{fig:fig4}, fails at reconstructing the spectrum of the original matrix $\mathbf{Q}(k)$ over the whole interval $[0,k_c]$, and does not detect the critical value $k_c$ either.

\subsection{The Fluctuation-Dissipation Theorem reloaded}
\label{sec:fdt2}
White noise is introduced in the reduced dynamics via, again, the Fluctuation-Dissipation Theorem.
We thus construct a Markov process for the variables $\{y_1,y_2\}$ (which now replace $\{x_1,x_2\}$, respectively) in the form:
\begin{subequations}
\label{red3}
\begin{align}
 \dot{y}_1&=a y_1+ b y_2 + \sigma_{11}^{(r)} \dot{W}_1+\sigma_{12}^{(r)} \dot{W}_2 \, , \label{red3a}
\\
\dot{y}_2&=c y_1+ d y_2 + \sigma_{21}^{(r)} \dot{W}_1+\sigma_{22}^{(r)} \dot{W}_2 \, , \label{red3b}   
\end{align}
\end{subequations}
where $\{a,b,c,d\}$ are the elements of the matrix $\mathbf{M}$ in \eqref{M2}, while the yet unknown parameters $\sigma_{i,j}^{(r)}$, $i,j=1,2$ enter the definition of the diffusion matrix $ \mathbf{\Sigma}=\frac{1}{2}\boldsymbol{{\sigma}}^{(r)} \left(\boldsymbol{{\sigma}}^{(r)}\right)^T$. Adopting the notation of Sec. \ref{sec:fdt}, the elements of the covariance matrix $\mathbf{R}(t,s)$ are now defined as:
\begin{equation}
R_{ij}(t,s):=\langle 
%(y_i(t)-y_i(0))(y_j(s)-y_j(0) 
\delta y_i(t) ~ \delta y_j(s)
\rangle \; , \; i=1,2 \; , 
\end{equation}
and we also set
\begin{equation}
  \overline{R}_{ij}:=\lim_{t\rightarrow \infty} R_{i,j}(t,t) \;. \label{steady2} 
\end{equation}
We then exploit the Lyapunov equation \eqref{lyap} to determine the elements of the diffusion matrix $\mathbf{\Sigma}$. %\cite{Pavl}:
%\begin{equation}
% \mathbf{M}\overline{\mathbf{R}}+\overline{\mathbf{R}} \mathbf{M}^T=-\mathbf{\Sigma} \label{lyap} \; ,
%\end{equation}
%which is an instance of the Fluctuation-Dissipation theorem of the II kind \cite{Zwanzig}.
To this aim, we must require that the elements of the stationary covariance matrix $\overline{\mathbf{R}}$ in Eq. \eqref{steady2} coincide with those of the corresponding covariance matrix evaluated through the original dynamics. Namely, we set:
\begin{equation}
   \overline{R}_{ij}= \lim_{t\rightarrow\infty}\langle 
%   (x_i(t)-x_i(0))( x_j(t)-x_j(0))
\delta x_i(t) ~ \delta x_j(t)
\rangle \label{equiv2} \; .
\end{equation}
%Hence, since the matrix $\overline{\mathbf{R}}$ is obtained from the original dynamics, while the matrix $\mathbf{M}$ has been characterized via the Invariant Manifold method, the unknown matrix $\mathbf{\Sigma}$ can thus be determined by solving Eq. \eqref{lyap}.
An explicit representation of $\mathbf{\Sigma}$ can be provided, also in this case, for identical beads, namely with $\gamma_1=\gamma_2=\gamma$, $\omega_1=\omega_2=\omega$ and $\sigma_1=\sigma_2=\sigma=\sqrt{2\gamma/(\beta m)}$.
Repeating the derivation outlined in Sec. \ref{sec:fdt}, we obtain:
\begin{subequations}
\label{x1x2}
\begin{align}
%x_1(t)&=& \frac{1}{2}\left(A_- e^{-\lambda_- t} + A_+ e^{-\lambda_+ t}+ B_- e^{-\phi_- t} + B_+ e^{-\phi_+ t}+\right. \nonumber\\
%&& \left. a_- e^{-\lambda_- t}+a_+ e^{-\lambda_+ t}+b_- e^{-\phi_- t}+b_+ e^{-\phi_+ t}\right) \; , \label{x1b2}  \\
%x_2(t)&=& \frac{1}{2}\left(A_- e^{-\lambda_- t} + A_+ e^{-\lambda_+ t} - B_- e^{-\phi_- t} - B_+ e^{-\phi_+ t}+\right. \nonumber\\
%&& \left. a_- e^{-\lambda_- t}+a_+ e^{-\lambda_+ t}-b_- e^{-\phi_- t}-b_+ e^{-\phi_+ t}\right)  \; . \label{x2b} 
x_1(t)&= \frac{1}{2}\left(A_- e^{\lambda^- t} + A_+ e^{\lambda^+ t}+ B_- e^{\phi^- t} + B_+ e^{\phi^+ t}\right. \nonumber\\
&+\left. a_- e^{\lambda^- t}+a_+ e^{\lambda^+ t}+b_- e^{\phi^- t}+b_+ e^{\phi^+ t}\right) \; , \label{x1b2}  \\
x_2(t)&= \frac{1}{2}\left(A_- e^{\lambda^- t} + A_+ e^{\lambda^+ t} - B_- e^{\phi^- t} - B_+ e^{\phi^+ t}\right. \nonumber\\
&+ \left. a_- e^{\lambda^- t}+a_+ e^{\lambda^+ t}-b_- e^{\phi^- t}-b_+ e^{\phi^+ t}\right)  \; , \label{x2b} 
\end{align}
\end{subequations}

and finally arrive at the following expressions:
\begin{subequations}
\label{Rij2}
\begin{align}
%\overline{R}_{11}(k)=\overline{R}_{22}(k)&=&\frac{\sigma^2}{4}\left[\frac{1}{(\lambda_+-\lambda_-)^2}\left(\frac{1}{\lambda_-}-\frac{4}{\lambda_+ \lambda_-}+\frac{1}{\lambda_+}  \right) \nonumber \right. \\
%&+& \left.\frac{1}{(\phi_+-\phi_-)^2}\left(\frac{1}{\phi_-}-\frac{4}{\phi_+ \phi_-}+\frac{1}{\phi_+}  \right)
%\right] \; , \label{R11b} \\
%\overline{R}_{12}(k)=\overline{R}_{21}(k)&=&\frac{\sigma^2}{4}\left[\frac{1}{(\lambda_+-\lambda_-)^2}\left(\frac{1}{\lambda_-}-\frac{4}{\lambda_+ \lambda_-}+\frac{1}{\lambda_+}  \right) \nonumber \right. \\
%&-& \left.\frac{1}{(\phi_+-\phi_-)^2}\left(\frac{1}{\phi_-}-\frac{4}{\phi_+ \phi_-}+\frac{1}{\phi_+}  \right)
%\right] \;. \label{R12b}
\overline{R}_{11}(k)=\overline{R}_{22}(k)&=\frac{\sigma^2}{4}\left[\frac{1}{(\lambda^- -\lambda^+)^2}\left(-\frac{1}{\lambda^-}+\frac{4}{\lambda_+ + \lambda_-}-\frac{1}{\lambda^-}  \right) \nonumber \right. \\
&+ \left.\frac{1}{(\phi^- -\phi^+)^2}\left(-\frac{1}{\phi^-}+\frac{4}{\phi^+ +\phi^-}-\frac{1}{\phi^-}  \right)
\right] \; , \label{R11b} \\
\overline{R}_{12}(k)=\overline{R}_{21}(k)&=\frac{\sigma^2}{4}\left[\frac{1}{(\lambda^- -\lambda^+)^2}\left(-\frac{1}{\lambda^-}+\frac{4}{\lambda^+ + \lambda^-}-\frac{1}{\lambda^+}  \right) \nonumber \right. \\
&- \left.\frac{1}{(\phi^- -\phi^+)^2}\left(-\frac{1}{\phi^-}+\frac{4}{\phi^+ + \phi^-}-\frac{1}{\phi^+}  \right)
\right] \;. \label{R12b}
\end{align}
\end{subequations}
It should be noted that Eqs. \eqref{R11b}-\eqref{R12b} imply  $\overline{R}_{11}(0)=\overline{R}_{22}(0)=(\beta m \omega^2)^{-1}$ and $\overline{R}_{12}(0)=\overline{R}_{21}(0)=0$, thus recovering the results known for the single-oscillator model.
%as also shown in Fig. \ref{fig:fig5}.
From Eq. \eqref{lyap}, one obtains the diffusion matrix $\mathbf{\Sigma}$, whose elements read:
\begin{subequations}
\label{sigm}
\begin{align}
\Sigma_{11}(k)&= - \left(a \overline{R}_{11} + b \overline{R}_{12}\right) \; , \label{sigm11}\\
\Sigma_{12}(k)=\Sigma_{21}(k)&= 
-\frac{1}{2}\left[(b + c) \overline{R}_{11} + (a + d) \overline{R}_{12}\right] \; ,\label{sigm12}\\
\Sigma_{22}(k)&= - \left(d \overline{R}_{11} + c \overline{R}_{12}\right) \label{sigm22} \; .
\end{align}
\end{subequations}
%The left panel of Fig. \ref{fig:fig5} shows the behavior of the real part of the coefficients $a_*,\dots,d_*$. Remarkably, in the absence of coupling it holds $a_*(0)=d_*(0)=\lambda_-$, namely the two coefficients recover the value of the drift of an overdamped single oscillator model (denoted by a red disk in the left panel) . Analogously, the red panel of Fig. \ref{fig:fig5} shows the behavior of the elements of the diffusion matrix $\mathbf{\Sigma}$ as functions of $k$. The limiting value $\Sigma_{11}(0)=\Sigma_{22}(0)=\lambda_-/(\beta m \omega^2)$ corresponds to the exact diffusion coefficient of an overdamped single oscillator model, cf. Ref. \cite{ColMun22}.
The behavior of the real part of $\Sigma_{11}(k)=\Sigma_{22}(k)$ and $\Sigma_{12}(k)=\Sigma_{21}(k)$ is shown in Fig. \ref{fig:fig5}, which highlights some interesting features induced by the coupling in the reduced dynamics defined in Eq. \eqref{red3}. In fact, not only the limiting value $\Sigma_{11}(0)=\Sigma_{22}(0)=\lambda_-/(\beta m \omega^2)$ recovers the exact diffusion coefficient of a single overdamped oscillator model obtained in \cite{ColMun22}, but for values of $k$ in the interval $[0,k_c]$, there also appears a \textit{negative} cross-diffusion coefficient $\Sigma_{12}=\Sigma_{21}$. We also remark that the diffusion matrix $\mathbf{\Sigma}$, evaluated through the set $\{a_*,b_*,c_*,d_*\}$ solving Eqs. \eqref{IE1b}-\eqref{IE4b}, is symmetric and positive semidefinite, since $\Sigma_{ii}>0$, $i=1,2$ and its determinant $Det[\mathbf{\Sigma}]\ge 0$.

\begin{figure}[ht]
     \centering
%       \includegraphics[width=0.45\textwidth]{fig5a}
%       \hskip 20pt
       \includegraphics[width=0.6\textwidth]{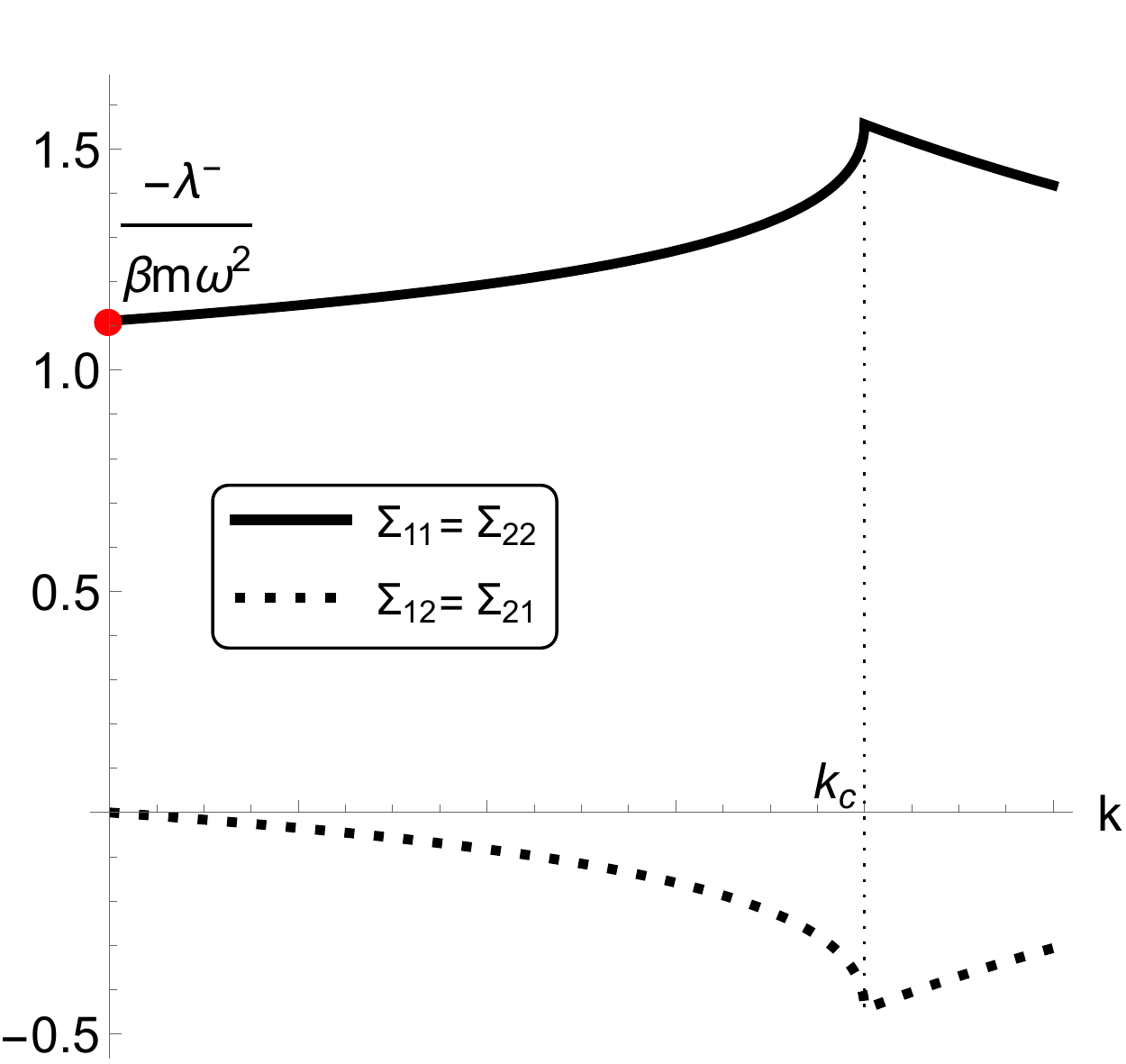}       
        \caption{Behavior of the real part of $\Sigma_{12}$ and $\Sigma_{22}$ as functions of $k$ for identical bead for the reduced dynamics defined in Eq. \eqref{red3}. The red disk at $k=0$ corresponds to the limiting value $-\lambda^{-}/\beta m \omega^2$. The parameters are set to the same values as in Fig. \ref{fig:fig3}.}
        \label{fig:fig5}
\end{figure}

\section{A single overdamped oscillator}
\label{sec:commute}

Let us now attempt one final step along our procedure of model reduction. Turning back, for a moment, at the reduction diagram shown in Fig. \ref{fig: diagram}, we observe that we have pursued, so far, two alternative paths,
corresponding to the top left and top right arrows shown in the diagram. 
In the first case, considered in Sec. \ref{sec:IM}, we eliminated the variables $\{x_2,v_2\}$ describing the second oscillator (top left arrow); in the second case, tackled in Sec. \ref{sec:IM2}, we removed both velocities $\{v_1,v_2\}$ from the original description, and obtained a system of coupled overdamped oscillators (top right arrow).
One may proceed even further at this stage, and construct a lower dimensional description in terms of the only variable $x_1$. This operation can be done by either removing the velocity $v_1$ from one model, or the position $x_2$ from the other (bottom left and bottom right arrows, respectively).
%Let us shortly summarize what we have accomplished so far.
%We started from the so-called original dynamics, expressed by the set of Eqs. \eqref{original}: we refer to the corresponding level of description as the level ``A''.
%Our reduction scheme then proceeded along two alternative routes, each featuring a specific set of ``slow'' variables. In the first instance, we erased the variables $(x_2,v_2)$ of the second oscillator and constructed an effective dynamics for the variables $(x_1,v_1)$. We denote such contracted description as level ``B1''.
%Next, starting from level ``B1'', one may then proceed further and obtain an overdamped description for the first oscillator in terms of the single variable $x_1$. The resulting overdamped description is  referred \AM{in this framework as} level ``C1'' %.
In either case, we target a reduced \AM{description} for the variable $z_1\in\mathbb{R}$ (here replacing $x_1$) in the form:
\begin{equation}
    \dot{z}_1=\alpha(k) z_1 + \sqrt{2 \mathcal{D}_r(k)} \dot{W}_1
    \label{red4} \; ,
\end{equation}
where the real-valued coefficients $\alpha$ and $\mathcal{D}_r$ depend on the various parameters of the model and, specifically, on the coupling parameter $k$.
%In fact, an analogous equation may be also derived from Eq. \eqref{red3} - corresponding to level ``B2'' - by erasing the configuration variable of the second oscillator in Eq. \eqref{red3}. That would give hence rise to a reduced description denoted as level ``C2''. 
\AM{A natural question thus arises:  do the parameters $\alpha(k)$ and $\mathcal{D}_r(k)$, in Eq. \eqref{red4}, depend on the chosen reduction path, or are they instead uniquely determined by the attained level of description?} Thus, the question is whether the two paths denoted by the left-sided and right-sided arrows in Fig. \ref{fig: diagram} commute.  We shall answer positively this question by explicitly evaluating the coefficients $\alpha$ and $\mathcal{D}_r$ along the two alternative paths.

Let us first address one of the two aforementioned reduction steps, e.g. the one denoted by the bottom left arrow in the diagram. Thus, starting from Eqs. \eqref{red2}, one finds that $\alpha(k)$ equals the eigenvalue of the matrix $\mathbf{M}$ endowed with the minimal non-zero absolute value, i.e.:
\begin{equation}
\alpha(k)=\xi^-(k) \label{xi1} \; .
\end{equation}
Therefore, owing to Eq. \eqref{req}, $\alpha$ also has the following asymptotics:
\begin{equation}
    \lim_{k\rightarrow 0}\alpha(k)=\lambda_1^{-} \label{alpha} \; .
\end{equation}
which recovers the result outlined in Ref. \cite{ColMun22}.
We remark that the function $\alpha(k)$ fully matches the branch of the spectrum of the original matrix $\mathbf{Q}(k)$ sharing the same asymptotic behavior. 
The bottom right arrow in the diagram of Fig. \ref{fig: diagram} amounts, instead, to looking for an overdamped dynamics for a single oscillator obtained from Eqs. \eqref{red3}. In this case, after some straightforward algebra, one finds that $\alpha(k)$ coincides with one of the two eigenvalues of the matrix $\mathbf{P}$, i.e.:
\begin{equation}
\alpha(k)\in \{\zeta^{+},\zeta^-\} \; . \label{alpha2}
\end{equation}
Next, enforcing the asymptotics \eqref{alpha} makes it possible to identify $\alpha(k)$ with the meaningful branch of the spectrum of the matrix $\mathbf{Q}(k)$. Hence, both derivations eventually lead to the same result: $\alpha(k)$ coincides with the eigenvalue of $\mathbf{Q}(k)$ reducing to $\lambda_1^-$ as $k\rightarrow 0$.
Next, the coefficient $\mathcal{D}_r$ can be determined in both cases using the Fluctuation-Dissipation relation. Namely, one invokes the scale invariance of the stationary covariance matrix, which in this case amounts to requiring 
\begin{equation}
    \lim_{t\rightarrow\infty} \langle %(z_1(t)-z_1(0))^2\rangle = \overline{R}_{11} \;
(\delta z_1(t))^2\rangle = \overline{R}_{11} \;
    , \label{z1sq}
\end{equation}
with $\overline{R}_{11}$ given in Eqs. \eqref{R11} and \eqref{R11b} in the first and in the second derivation, respectively. Then, Eq. \eqref{lyap} finally leads to:
\begin{equation}
\mathcal{D}_r(k)= -\alpha(k) \overline{R}_{11} \; . \label{Dr}    
\end{equation}
It is worth pointing out that the same coefficient $\mathcal{D}_r(k)$ is obtained along both considered derivations owing to the uniqueness of $\alpha(k)$ as well as to the stipulated scale invariance of the stationary covariance matrix, which makes the two expressions of $\overline{R}_{11}$, in Eqs. \eqref{R11} and \eqref{R11b}, coincide.
%The \AM{right-hand side} of Eqs. \eqref{alpha2} corresponds to $|\lambda_1^{-}|$, thus recovering the identity with the \AM{right-hand side} of Eq. \eqref{alpha}.
%The coefficient $\mathcal{D}_r$ can be obtained following the same procedure \AM{as the one} leading to Eq. \eqref{Dr}. By Eqs. \eqref{equiv} and \eqref{equiv2}, and also using Eqs. \eqref{alpha} and \eqref{alpha2}, we conclude that the \AM{effective}  coefficient $\mathcal{D}_r$ \AM{(computed along the two alternative  paths)} is uniquely determined by Eq. \eqref{Dr}.

\section{Conclusions}
\label{sec: conclusion}
Reducing the analytical and computational complexity of many-particle systems is a central challenge in many scientific and engineering problems. In this paper we have developed self-consistent reduction schemes for linearly coupled Langevin equations that have been widely employed in the modelling and analysis of proteins, linear networks, and polymers \cite{Degond2009, Berkowitz1981,Soheilifard2011}. We started from a reference dynamics constituted by a classical toy model of non-equilibrium statistical mechanics, for which exact calculations can be carried out: two one-dimensional coupled Brownian oscillators. We then constructed a contracted Markovian description of the original model proceeding through two distinct steps. We first employed the invariant manifold method and the principle of Dynamic Invariance to obtain the deterministic component of the reduced dynamics. Then, we used the Fluctuation-Dissipation Theorem to compute the diffusion matrix. The proposed method displays several noteworthy features: it preserves the relevant branches of the spectrum of the original transport matrix as well as the relevant elements of the stationary %matrix characterizing the long-time behaviour of the 
covariance matrix, while it also fulfills, by construction, the Fluctuation-Dissipation Theorem expressed by a suitable Lyapunov equation. The application of the principle of Dynamic Invariance was shown, in \cite{GorKar05}, to lead to results that are equivalent to an exact summation of the Chapman-Enskog expansion. Here we show that the exact solution of the Invariance Equations, at variance with lower order approximations (such as the one yielding the weak coupling regime), witnesses the presence of a critical value of the coupling parameter which also occurs in the original dynamics. Another remarkable feature of our reduction method is that the transport and diffusion coefficients pertaining to a chosen level of description \textit{do not} depend on the reduction path: we highlighted this result through an explicit computation in Sec. \ref{sec:commute}. It would be interesting to reinterpret our results
  
We expect our method to be applicable to slightly more complicated coupled stochastic differential equations, such as models constituted by anharmonic chains of oscillators \cite{Eckmann1999,Eckmann2000}, or chains of oscillators of Kuramoto-Sakaguchi type, cf. \cite{Gottwald2015,SmithGottwald2020,YueLachlanGottwald2020}. Another promising application of our method concerns the reduction of an arbitrary nonlinear system interacting bilinearly with a harmonic oscillator heat bath, as discussed e.g. in \cite[Chapter~1.6]{Zwanzig}. As a direct follow-up of this work, other potentially interesting topics for future study include the generalization of our method to nonlinear reaction coordinates and the derivation of quantitative estimates for the practical control of approximation errors. 

\appendix
\label{sec:app}

\section{Roots of characteristic polynomials}
In this appendix, we provide detailed computations and proof for proposition \ref{prop: eigenvalues}. We start with an auxiliary well-known Lemma on the roots of a quartic equation, see for instance \cite{Rees1922}.
\begin{lemma}
Consider the general quartic equation
$$
\mathcal{Q}(x)=ax^4+bx^3+cx^2+dx+e=0,
$$
with \AM{given coefficients $a,b,c,d,e\in\mathbb{R}$ so that}  $a\neq 0$.
\AM{Furthermore, set}
\begin{subequations}
\label{eq: DeltaPD}
\begin{align}
\Delta&:=256 a^3 e^3-192 a^2bde^2-128 a^2c^2 e^2+144 a^2cd^2e-27a^2d^4+144 ab^2ce^2\notag
\\&\qquad-6ab^2d^2 e-80 abc^2de+18 abcd^3+16 ac^4 e-4ac^3d^2-27b^4e^2\notag
\\&\qquad\qquad+18 b^3cde-4b^3d^3-4b^2c^3e+b^2c^2d^2,\notag
\\ P&:=8ac-3b^2,\notag
\\ D&:= 64a^3e-16a^2c^2+16ab^2c-16a^2bd-3b^4.
\end{align}
\end{subequations}
Then \AM{the following statements hold true:}
\begin{itemize}
\item[(i)] $\mathcal{Q}$ has two distinct real roots and two complex conjugate non-real roots if and only if $\Delta<0$.
\item[(ii)] $\mathcal{Q}$ has four distinct real roots if and only if $\Delta>0, P<0, D<0$.
\item[(iii)] $\mathcal{Q}$ has two pairs of non-real complex conjugate roots if and only if $\Delta>0$ and either $P>0$ or $D>0$.
\end{itemize}
\end{lemma}
Using the above lemma, we next provide a proof for Proposition \ref{prop: eigenvalues}.
\begin{proof}[Proof of Proposition \ref{prop: eigenvalues}]
We recall that  
\begin{equation*}
    \mathbf{Q}(k)=\begin{pmatrix}
0 & 1 & 0 & 0 \\
-\omega_1^2-k & -\gamma_1 & k & 0 \\
0 & 0 & 0 & 1 \\
k & 0 & -\omega_2^2-k & -\gamma_2 
\end{pmatrix}.
\end{equation*}
The characteristic polynomial of $\mathbf{Q}(k)$ is given by
\begin{align}
\mathcal{Q}_k(\lambda)&=\lambda^4+(\gamma_1+\gamma_2)\lambda^3+(\omega_1^2+\omega_2^2+\gamma_1\gamma_2+2k)\lambda^2+(\gamma_1 \omega_2^2+\gamma_2\omega_1^2+(\gamma_1+\gamma_2)k)x\notag
\\&\qquad+\omega_1^2\omega_2^2+(\omega_1^2+\omega_2^2)k\notag
\\&=(\lambda^2+\gamma_1 \lambda+\omega_1^2)(\lambda^2+\gamma_2 \lambda+\omega_2^2)+ k\Big[(\lambda^2+\gamma_1 \lambda+\omega_1^2))+(\lambda^2+\gamma_2 \lambda+\omega_2^2)\Big]\notag
\\&= \mathcal{Q}_0(\lambda)+k\Big[(\lambda^2+\gamma_1 \lambda+\omega_1^2))+(\lambda^2+\gamma_2 \lambda+\omega_2^2)\Big],
\end{align}
where $\mathcal{Q}_0$ is the characteristic polynomial of $\mathbf{Q}_0$, \AM{that is}
$$
\mathcal{Q}_0(\lambda)=(\lambda^2+\gamma_1 \lambda+\omega_1^2)(\lambda^2+\gamma_2 \lambda+\omega_2^2).
$$
Let $\Delta_k, P_k, D_k$ be defined as in \eqref{eq: DeltaPD} for $\mathcal{Q}_k$. Then we have
\begin{equation}
\Delta_k=\Delta_0+ k \mathcal{P}_4(k),\quad P_k=P_0+16 k, \quad D_k=-64 k^2+16(\gamma_1-\gamma_2)^2 k+D_0,
\end{equation}
where $\mathcal{P}_4(k)$ is a polynomial of degree 4 of $k$ whose leading coefficient is given by
$$
p_4=256(\omega_1^2+\omega_2^2)-32(\gamma_1+\gamma_2)^2.
$$
Since $\gamma_i^2-4\omega_i^2>0$ for $i=1,2$, $\mathcal{Q}_0$ has 4 distinct real roots. Therefore
$$
\Delta_0>0, \quad P_0<0, \quad D_0<0.
$$
Since $\Delta_k$ and $D_k$ are continuous functions \AM{in the variable} $k$ and $\lim\limits_{k\rightarrow 0^+} \Delta_k=\Delta_0<0$ and $\lim\limits_{k\rightarrow 0^+} D_k=D_0<0$, for sufficiently small $k\leq k_1$, $\Delta_k$ and $D_k$ are both negative. Hence, for $k<\min\{k_1,-\frac{P_0}{16}\}$, we have
$$
\Delta_k<0, \quad P_k<0, \quad D_k<0.
$$
Therefore, $\mathcal{Q}_k$ has four distinct real roots.

For $k>-\frac{P_0}{16}$, then $P_k>0$. Since $\Delta_k$ is a polynomial with the leading coefficient $p_4=32[8(\omega_1^2+\omega_2^2)-(\gamma_1+\gamma_2)^2]$ then for sufficiently large $k$, $\Delta_k$ is positive if $8(\omega_1^2+\omega_2^2)> (\gamma_1+\gamma_2)^2$ and is negative if $8(\omega_1^2+\omega_2^2) < (\gamma_1+\gamma_2)^2$. Therefore for sufficiently large $k$, $\mathcal{Q}_k$ has two distinct real roots and two complex conjugate non-real roots if $8(\omega_1^2+\omega_2^2) < (\gamma_1+\gamma_2)^2$, and has two pairs of non-real complex conjugate roots if $8(\omega_1^2+\omega_2^2) > (\gamma_1+\gamma_2)^2$.

In the case of identical beads, we have
$$
\mathcal{Q}_k(\lambda)=(\lambda^2+\gamma \lambda+ \omega^2)(\lambda^2+\gamma \lambda+ \omega^2+2k)
$$ 
Since $\gamma^2-4\omega^2>0$, the quadratic $\lambda^2+\gamma \lambda+ \omega^2$ has two distinct roots. The second quadratic $\lambda^2+\gamma \lambda+\omega^2+2k=0$ has two distinct roots if $\gamma^2-4(\omega^2+2k)>0$, that is $k<k_c$, and has a pair of complex conjugate roots if $k>k_c$. 

This completes the proof of the lemma.
\end{proof}
We consider the following concrete example used in the previous sections.
\begin{example}
\label{ex: phase transition}
In particular, for the parameters' values in Figure \ref{fig:fig1} 
$$
\omega_1=0.3, \quad \omega_2=0.4, \gamma_1=1.0, \quad \gamma_2=1.2
$$
Then
\begin{align*}
& \mathcal{Q}_k(\lambda)=0.0144+0.25 k+(0.268+2.2 k) \lambda+(1.45 +2k) \lambda^2+ 2.2 \lambda^3+\lambda^4,
\\& \Delta_k=0.000015488-0.00039952 k+0.0384776 k^2-0.568736 k^3+9.7808 k^4-90.88 k^5,
\\& P_k=16k-2.92, \quad D_k=-64k^2+0.64 k-0.1408.
\end{align*}
The equation $\Delta_k=0$ has only one positive real root $k_c\approx 0.0876$. \AM{It holds:} 
\begin{itemize}
    \item[(i)] \AM{If} $k<k_c$, $\Delta_k>0, P_k<0, D_k<0$, \AM{then} $\mathcal{Q}_k$ has 4 distinct real roots.
    \item[(ii)] \AM{If} $k=k_c$, $\Delta_k=0$, $P_k<0, D_k<0$, \AM{then} $\mathcal{Q}_k$ has a real double root and two real simple roots (i.e., 3 distinct real roots).
    \item[(iii)] \AM{If} $k>k_c$, $\Delta<0$, \AM{then} $\mathcal{Q}_k$ has has two distinct real roots and two complex conjugate non-real roots.
\end{itemize}
\end{example}
\section*{Acknowledgment} MC thanks Lamberto Rondoni (Turin Polytechnic, Italy) for useful discussions. Research of MHD was supported by EPSRC Grants EP/W008041/1 and  EP/V038516/1.
\bibliographystyle{plain}
\bibliography{biblio}

\end{document}